\documentclass[a4paper,reqno,twoside]{amsart}

\usepackage[%
pdfauthor={Habib Ammari,%
            Silvio Barandun,%
            and Alexander Uhlmann},%
pdftitle={Truncated Floquet-Bloch transform for computing the spectral properties of large finite systems of resonators},
pdfkeywords={Truncated Floquet-Bloch transform, TFBT}
]{hyperref}

\usepackage[left=3.5cm,right=3.5cm,top=3.5cm,bottom=3.5cm,footskip=1.5cm,headsep=1cm,bindingoffset=0.cm,]{geometry}

\usepackage[utf8]{inputenc}
\usepackage[english]{babel}
\usepackage[T1]{fontenc} 
\usepackage{lmodern} 
\rmfamily 
\DeclareFontShape{T1}{lmr}{b}{sc}{<->ssub*cmr/bx/sc}{}
\DeclareFontShape{T1}{lmr}{bx}{sc}{<->ssub*cmr/bx/sc}{}

\numberwithin{equation}{section}
\usepackage{amsmath}
\usepackage{amssymb}
\usepackage{amsthm}
\usepackage{amsfonts}
\usepackage{mathtools}

\usepackage{mathdots}
\usepackage{caption}
\usepackage{subcaption}

\usepackage{dsfont} 
\usepackage[format=hang]{caption}
\usepackage[normalem]{ulem}

\usepackage{standalone}
\usepackage{graphicx}
\usepackage{imakeidx}
\makeindex
\usepackage{setspace}
\usepackage{appendix}
\usepackage{pdfpages}
\usepackage{upgreek}
\usepackage{tabulary}
\usepackage{booktabs}
\usepackage{extarrows}
\usepackage{tabularx}
\usepackage{float}
\restylefloat{table}
\usepackage{enumitem}
\usepackage{csquotes}
\usepackage{bm}
\usepackage{orcidlink}
\usepackage{lipsum}

\usepackage{tikz}
\usepackage{tikz-layers}
\usepackage{tikz-cd}
\usepackage[arrow, matrix, curve]{xy}
\usetikzlibrary{matrix,positioning,decorations.pathreplacing,calc}
\usetikzlibrary{
    decorations.text,%
    decorations.markings,%
    shadows}
\usetikzlibrary{calc,intersections}
\usepackage{adjustbox}
\usepackage{graphicx}

\usepackage[
    style=numeric,
    sorting=nty,
    maxnames=99,
    maxalphanames=5,
    natbib=true,
    backend=bibtex,
    sortcites]{biblatex}

\DeclareNameAlias{default}{family-given}

\AtEveryBibitem{
    \clearfield{url}
    \clearfield{issn}
    \clearfield{isbn}
    \clearfield{urldate}

    \ifentrytype{book}{
        \clearfield{pages}}{%
    }
}

\usepackage{xargs}                      
\usepackage[prependcaption,textsize=tiny,textwidth=3cm]{todonotes}
\newcommandx{\unsure}[2][1=]{\todo[linecolor=red,backgroundcolor=red!25,bordercolor=red,#1]{#2}}
\newcommandx{\change}[2][1=]{\todo[linecolor=blue,backgroundcolor=blue!25,bordercolor=blue,#1]{#2}}
\newcommandx{\info}[2][1=]{\todo[linecolor=OliveGreen,backgroundcolor=OliveGreen!25,bordercolor=OliveGreen,#1]{#2}}
\newcommandx{\improvement}[2][1=]{\todo[linecolor=black,backgroundcolor=black!25,bordercolor=black,#1]{#2}}
\newcommandx{\thiswillnotshow}[2][1=]{\todo[disable,#1]{#2}}

\usepackage[noabbrev,capitalize]{cleveref}
\crefname{proposition}{Proposition}{Propositions}
\crefname{equation}{}{}

\newtheorem{theorem}{Theorem}[section]
\newtheorem{lemma}[theorem]{Lemma}
\newtheorem{proposition}[theorem]{Proposition}
\newtheorem{corollary}[theorem]{Corollary}

\theoremstyle{definition}
\newtheorem{definition}[theorem]{Definition}
\newtheorem{example}[theorem]{Example}
\newtheorem{assumption}[theorem]{Assumption}
\newtheorem{remark}[theorem]{Remark}

\crefname{assumption}{Assumption}{Assumptions}
\crefname{definition}{Definition}{Definitions}
\crefname{corollary}{Corollary}{Corollaries}
\crefname{enumi}{item}{items}

\usepackage[final]{microtype}

\makeatletter
\newsavebox\myboxA
\newsavebox\myboxB
\newlength\mylenA

\newcommand*\xoverline[2][0.75]{%
  \sbox{\myboxA}{$\m@th#2$}%
  \setbox\myboxB\null
  \ht\myboxB=\ht\myboxA%
  \dp\myboxB=\dp\myboxA%
  \wd\myboxB=#1\wd\myboxA
  \sbox\myboxB{$\m@th\overline{\copy\myboxB}$}
  \setlength\mylenA{\the\wd\myboxA}
  \addtolength\mylenA{-\the\wd\myboxB}%
  \ifdim\wd\myboxB<\wd\myboxA%
    \rlap{\hskip 0.5\mylenA\usebox\myboxB}{\usebox\myboxA}%
  \else
    \hskip -0.5\mylenA\rlap{\usebox\myboxA}{\hskip 0.5\mylenA\usebox\myboxB}%
  \fi}
\makeatother

\usepackage{fancyhdr}

\fancypagestyle{plain}{%
    \fancyhead[C]{}
    \fancyfoot[C]{}
    
    }
\fancypagestyle{nosection}{%
    \fancyhead[CE]{}
    \fancyhead[CO]{}
    \fancyfoot[CE,CO]{\thepage}
}

\newcommand{\TO}{\mathbb{T}^1}

\DeclareMathOperator{\argmax}{argmax}

\newcommand{\setsymbols}{C_{k\times k}}

\newcommand{\rapprox}[2]{#1^{[#2]}}

\newcommand{\tftn}{truncated Floquet-Bloch transform~}

\newcommand{\dfrt}{\mc{F}}

\usepackage{tikz}

\usetikzlibrary{matrix} 
\usetikzlibrary{arrows,arrows.meta,bending} 

\usepackage{nicematrix}

\newcommand{\expi}[1]{e^{\i#1}}

\usetikzlibrary{hobby}

\DeclareMathOperator{\N}{\mathbb{N}}
\DeclareMathOperator{\Z}{\mathbb{Z}}

\DeclareMathOperator{\R}{\mathbb{R}}
\DeclareMathOperator{\C}{\mathbb{C}}

\renewcommand{\i}{\mathbf{i}}
\renewcommand{\tilde}{\widetilde}

\renewcommand{\qedsymbol}{$\blacksquare$}

\DeclareMathOperator{\diag}{diag}

\renewcommand{\epsilon}{\varepsilon}

\let\emptyset\varnothing

\renewcommand{\i}{\mathbf{i}}
\renewcommand{\tilde}{\widetilde}

\usepackage{tcolorbox}
\usetikzlibrary{tikzmark,calc}

\newcommand{\ip}[2]{\left\langle #1, #2 \right\rangle}
\newcommand{\mc}[1]{\mathcal{#1}}

\newcommand{\abs}[1]{\left\lvert#1\right\rvert}

\newcommand{\norm}[1]{\left\lVert#1\right\rVert}

\newcommandx{\silvio}[2][1=]{\todo[linecolor=blue,backgroundcolor=blue!25,bordercolor=blue,#1]{Silvio: #2}}

\title[Truncated Floquet-Bloch Transform]{Truncated Floquet-Bloch transform for computing the spectral properties of large finite systems of resonators}

\begin{document}
 \author[H. Ammari]{Habib Ammari \,\orcidlink{0000-0001-7278-4877}}
\address{\parbox{\linewidth}{Habib Ammari\\
 ETH Z\"urich, Department of Mathematics, Rämistrasse 101, 8092 Z\"urich, Switzerland, \href{http://orcid.org/0000-0001-7278-4877}{orcid.org/0000-0001-7278-4877}}.}
\email{habib.ammari@math.ethz.ch}
\thanks{}

\author[S. Barandun]{Silvio Barandun\,\orcidlink{0000-0003-1499-4352}}
 \address{\parbox{\linewidth}{Silvio Barandun\\
 ETH Z\"urich, Department of Mathematics, Rämistrasse 101, 8092 Z\"urich, Switzerland, \href{http://orcid.org/0000-0003-1499-4352}{orcid.org/0000-0003-1499-4352}}.}
 \email{silvio.barandun@sam.math.ethz.ch}

\author[A. Uhlmann]{Alexander Uhlmann\,\orcidlink{0009-0002-0426-6407}}
 \address{\parbox{\linewidth}{Alexander Uhlmann\\
 ETH Z\"urich, Department of Mathematics, Rämistrasse 101, 8092 Z\"urich, Switzerland, \href{http://orcid.org/0009-0002-0426-6407}{orcid.org/0009-0002-0426-6407}}.}

\email{alexander.uhlmann@sam.math.ethz.ch}

\maketitle

\begin{abstract}
The truncated Floquet-Bloch transform can be used to characterise the spectral properties of finite periodic and aperiodic large systems of resonators. This paper aims to provide for the first time the mathematical foundations of this transform.
\end{abstract}

\bigskip

\par

\bigskip

\\

\section{Introduction}

The spectral properties of large finite systems of subwavelength resonators are a difficult and outstanding problem in wave physics and metamaterial science. Understanding the nature and formation of bandgaps and localised modes associated with wave propagation in finite periodic or aperiodic structures presents fundamental and technological challenges. In particular, basic questions are about the existence of bandgaps and the mechanism of localisation. 

To explore these issues, a \emph{capacitance matrix} formulation has been derived to approximate the eigenmodes and eigenfrequencies of systems of subwavelength resonators in \cite{ammari.davies.ea2021Functional,cbms}. Then a numerical method called \emph{truncated Floquet-Bloch transform} has been introduced in \cite{ammari.davies.ea2023Convergence,ammari.davies.ea2023Spectral}. The truncated Floquet-Bloch transform has been used to approximate the spectral properties of large finite subwavelength resonator systems, giving a concrete way to associate bandgaps and localised modes to practically realisable materials that may be finite periodic or aperiodic. 
The main idea is that when the size of the finite structure is sufficiently large, the structure's eigenmodes are approximately a linear combination of Bloch modes of the corresponding infinite structure. To compare the discrete eigenvalues of the finite structure to the continuous spectrum of the infinite periodic structure, one can reverse engineer the appropriate quasiperiodicities corresponding to these Bloch modes, taking into account the symmetries in the problem.  The discrete band function and defect mode calculations in  \cite{ammari.davies.ea2023Convergence,ammari.davies.ea2023Spectral} provide a notion of how an eigenmode of the finite problem is approximated either by Bloch modes (delocalised eigenmodes) or defect modes (localised eigenmodes, corresponding to defect eigenfrequencies inside a bandgap) of the infinite structure. Localised modes can be obtained by (appropriately) changing  the material parameters inside the resonators 
\cite{ammari.davies.ea2024Anderson} or introducing structural defects \cite{ammari.davies.ea2020Topologically}. 

In this work, our aim is to provide for the first time the mathematical foundations of the truncated Floquet-Bloch transform, which is a tool able to recover the band structure from a finite system. 
After reformulating the original definition introduced in \cite{ammari.barandun.ea2024Spectra} to make it more stable and more efficient, we first consider truncated periodic systems of resonators. We explore two cases of interest: sets of resonators repeated periodically with either only nearest neighbour interactions or with long-range interactions. Based on $k$-Toeplitz theory and the fact from \cite{ammari.davies.ea2023Spectral} that any eigenmode of the infinitely periodic structure has a corresponding approximation in the large, finite counterpart, we prove (in \cref{thm: TFT general}) that the truncated Floquet-Bloch transform can recover the limiting quasiperiodicity from the bulk eigenmodes of the finite system. In the nearest-neighbour and $k=1$ case, the proof turns out to be explicit (\cref{prop: DFT gives two diracts on tridiagonal toeplitz and 1D cap mat}). This result is then generalised to systems of resonators with long-range interactions by approximating the full capacitance matrix by a banded Toeplitz matrix (see again \cref{thm: TFT general}).

Subsequently, we apply the truncated Floquet-Bloch transform to three aperiodic structures: Su--Schrieffer--Heeger (SSH) finite chains of resonators, compactly perturbed finite periodic systems, and dislocated chains of resonators. SSH chains consist of systems of dimers with a defect in the geometric structure, such that at some point the repeating pattern of alternating separation distances is broken. While the SSH model is a canonical example of a topologically protected interface mode \cite{original_ssh,ssh3d,ammari.barandun.ea2024Exponentially,fefferman1}, \emph{i.e.}, robust with respect to imperfections in the design, defect eigenmodes induced by compact perturbations of the system are known to be very sensitive to imperfections in the design \cite{ammari.davies.ea2021Functional,cbms}. Dislocated chains of resonators are one-dimensional arrays of pairs of subwavelength resonators where we introduce a defect by adding a dislocation within one of the resonator pairs. 
This dislocation generates a topologically protected edge mode \cite{ammari.davies.ea2022Robust,hempel.kohlmann2011variational,drouot.fefferman.ea2020Defect}.

For each of these three aperiodic structures, we show that the truncated Floquet-Bloch transform recovers the band structure and the localised eigenmodes in the bandgap.   Our results are illustrated by several numerical experiments. 

The paper is organised as follows. In \cref{sec:2}, we outline the motivation and main results of this work. \cref{sec:3} is devoted to fixing the notation and collecting some useful results in Toeplitz theory, which are the building blocks of our theory. In \cref{sec: DFT}, we introduce the truncated Floquet-Bloch transform and show how it interacts with the eigenvectors and pseudo-eigenvectors of circulant matrices. In \cref{sec:5}, we apply the truncated Floquet-Bloch transform to finite, periodic systems and prove that it enables the accurate recovery of their associated band structures. \cref{sec:6} is devoted to SSH chains of resonators, compactly defected structures, and dislocated chains of resonators for which the truncated Floquet-Bloch transform recovers both the band structure and the localised eigenmodes in the bandgap accurately. The paper ends with concluding remarks and discussions on the possible generalisation of the proposed method for more general disordered systems of resonators than those discussed here.

\section{Motivation and main results} \label{sec:2}
The Floquet-Bloch theorem is one of the most used tools in the study of infinite systems. Considering a lattice $\Gamma\subset \R^d$ with fundamental domain $Y$ and a function $f(x)$ decaying sufficiently fast, then the Floquet-Bloch transform of $f$ is defined as
\begin{align}
    \label{def:floquet-bloch-transform}
    \mathcal{U}[f](x,\alpha)\coloneqq \sum_{m\in\Gamma} f(x-m)e^{\i \alpha\cdot m},  \quad x,\alpha \in \R^d.
\end{align}
The Floquet-Bloch transform $\mathcal{U}[f](x,\alpha)$ is periodic in $\alpha$ and $\alpha$-quasiperiodic in $x$, that is, $\mathcal{U}[f](x+m,\alpha)=e^{\i \alpha\cdot m}\mathcal{U}[f](x,\alpha)$ for $m\in\Gamma$. It is therefore enough to know the function $\mathcal{U}[f](x,\alpha)$ in the fundamental domain $Y$ in the $x$-variable  and in the \emph{first Brillouin zone} $Y^*\coloneqq \mathbb{R}^d/\Gamma^*$ in the $\alpha$-variable. Here, $\Gamma^*$ is the dual lattice of $\Gamma$. 

The Floquet-Bloch theorem states that if $L$ is a self-adjoint differential operator with coefficients that are periodic with respect to the lattice $\Gamma$, then its spectrum is decomposed as
\begin{align}
\label{eq: floquet-bloch-theorem}
    \sigma(L) = \bigcup_{\alpha \in Y^*} \sigma(L(\alpha)), 
\end{align}
where $L(\alpha)$ is the original operator restricted to act on functions defined on the torus $Y$. Remark that if $L$ is elliptic, the operators $L(\alpha)$ have compact resolvents and hence discrete spectra, say $\mu_1(\alpha)\leq\mu_2(\alpha)\leq \dots$ counted with their multiplicities. Then \eqref{eq: floquet-bloch-theorem} becomes \begin{align}
\label{eq: spectrum differntial operator band functions}
    \sigma(L) = \bigcup_{n\in\N} \bigcup_{\alpha \in Y^*} \mu_n(\alpha).
\end{align}

The study of infinite periodic systems is therefore widespread in the literature. Nevertheless, they are mostly nonphysical and just a mathematical approximation of large, finite systems. However, in this paper, we show that the band functions $\mu_n(\alpha)$ are very much physical and represent the limit of their discrete equivalent as the size of the finite system grows.

\begin{figure}[h]
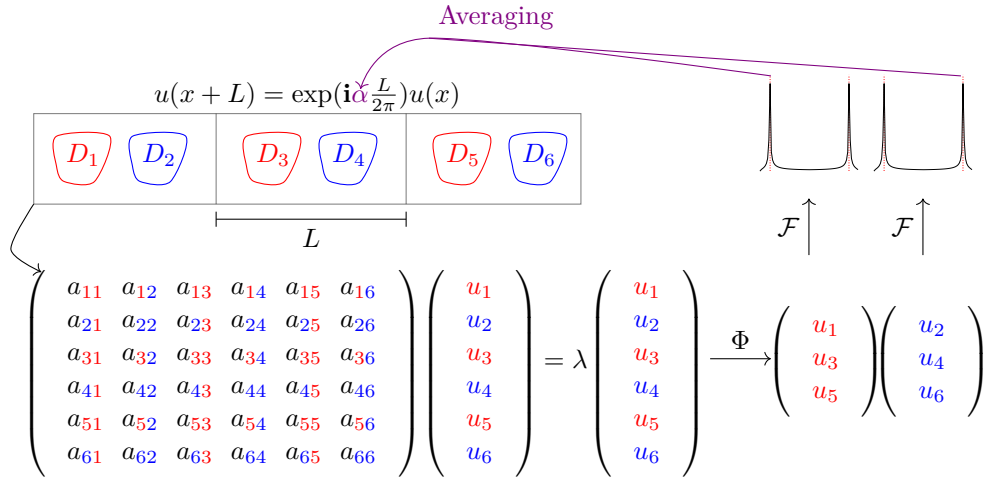

    \centering
    \include{plots/DFT.tex}
    \caption{Illustration of different steps behind the truncated Floquet-Bloch transform in reconstructing the quasiperiodicity of an eigenmode. The map $\mathcal{F}$ is the discrete  Fourier transform and $\Phi$ is defined in \eqref{eq: phi map slicing}.}
    \label{fig: illustrtion TFBT}
\end{figure}

We consider linear eigenvalue problems as the discrete equivalent of partial differential equation (PDE) problems. These discrete models are in particular used in practice in photonics and phononics \cite{ammari.davies.ea2021Functional,cbms} to approximate the resonant eigenfrequencies and eigenmodes of systems of high contrast, subwavelength resonators. We refer the reader to \cite{ammari.davies.ea2021Functional,cbms} for their derivations from continuous PDE  models. Analogous models appear in condensed matter models
under the tigh-binding approximation. 
We restrict our analysis to periodicity with respect to a one-dimensional lattice. 

Our approach is illustrated in \cref{fig: illustrtion TFBT}. The eigenmodes of a finite system of dimers can be modelled by a linear eigenvalue problem (we will see that the involved matrix has a Toeplitz structure due to the spatial periodicity). We regroup the entries of the eigenvector by the position in the unit cell. These vectors present a quasiperiodic behaviour (expected by the periodicity of the problem) whose quasiperiodicity can be extracted through the discrete Fourier transform (DFT). This quasiperiodicity is asymptotically the same as the one of the mode propagating through the infinite system with the same eigenfrequency.

The main result of this paper is \cref{thm: TFT general}, which shows that the procedure outlined in \cref{fig: illustrtion TFBT} recovers the quasiperiodicities of the eigenmodes associated to large, finite systems.

\section{Notation and Toeplitz theory} \label{sec:3}
\subsection{Toeplitz operators and matrices}
Let $k\in\N$ and consider a complex matrix-valued function  $f:\mathbb{T}^1\to \C^{k\times k}$ on the unit circle $\mathbb{T}^1\subset \C$ represented by its Fourier series
\begin{align}
    f(z)=\sum_{s\in\Z} a_s z^s,
\end{align}
where $a_s\in\C^{k\times k}$. 

The \emph{$k$-Toeplitz operator} associated to $f$ as above is the operator on $\ell^2(\N)$, the set of square summable sequences, given by 
$$
    T(f)=\begin{pmatrix}
         a_0     & a_{-1} & a_{-2}  & \cdots  \\[1mm]
            a_1     & a_0    & a_{-1} &   \cdots    \\[-1mm]
               a_2      & a_1 &  \ddots & \ddots  \\
               \vdots & \ddots & \ddots & \ddots 
    \end{pmatrix}.
$$

Let now $n, k\geq 1$ and define the projections
\begin{align}
    P_n:\ell^2(\N) & \to\ell^2(\N)\nonumber         \\
    (x_1,x_2,\dots)   & \mapsto(x_1,\dots x_n,0,0,\dots).
    \label{eq: definiton of Pn}
\end{align}
Then, the \emph{$k$-Toeplitz matrix} of order $mk$ for $m\in\N$ associated to the symbol $f$ is given by
$$
    T_{mk}(f)\coloneqq P_{mk}T(f)P_{mk}.
$$
$T_{m k}(f)$ can be identified as an $mk\times mk$ matrix.

\begin{remark}
    When $k=1$, $k$-Toeplitz matrices and operators are simply known as Toeplitz matrices and operators.
\end{remark}

\begin{definition}[Band functions]
\label{def: band functions}
    Consider a Hermitian Toeplitz operator $T(f)$ with continuous symbol $f:\TO\to \C^{k\times k}$. Then, its essential spectrum is given by $$\sigma_{ess}(T(f)) = \bigcup_{e^{\i\alpha}\in \TO}\sigma(f(e^{\i\alpha})).$$ This follows from \cite[Theorem 6.5]{bottcher.silbermann1999Introduction}). We can then continuously label the eigenvalues of $f(\alpha)$ to obtain the \emph{band functions} $\lambda_p:\TO \to \R$ for $p=1,\dots,k$. 
    
    For these bands, we also define the corresponding \emph{eigenvector band functions} $\bm u_p:\TO \to \C^k$ mapping $\expi{\alpha}$ to the eigenvector corresponding to the eigenvalue $\lambda_p(\expi{\alpha})$. We make $\bm u_p(\expi{\alpha})$ unique by requiring it be unit length and fixing a polarity, \emph{i.e.},  $\Im \bm u_p(\expi{\alpha})^{(0)} = 0$.
\end{definition}
Remark the similarity between \cref{def: band functions} and \eqref{eq: spectrum differntial operator band functions}.

\begin{assumption}\label{assumptions}
    For the entirety of the work, we make the following assumptions:
    \begin{enumerate}
        \item[(i)] We consider symbols with enough regularity, that is, $f:\mathbb{T}^1\to \C^{k\times k}$ continuous and piecewise differentiable, and whose image are Hermitian matrices, \emph{i.e.}, $f(e^{\i\alpha})_{i,j}=\overline{f(e^{\i\alpha})_{j,i}}$. 
        \item[(ii)] Eigenvectors are considered to be $\ell^2$-normalised and the eigenvalues are sorted so that $\lambda_i\leq \lambda_j$ for $i<j$.
        \item[(iii)] No band crossings occur, that is the images of the band functions $\lambda_p:\TO\to\R$ are disjoint for all $p$.
    \item[(iv)] $\lambda_p:\TO\to\R$ are continuously differentiable away from $0$ for all $p$.
    \item[(v)] There are no van der Hove points, that is, $\lambda'_p(e^{\i\alpha})\neq 0$ for any $\expi{\alpha}\in\TO\setminus\{\pm 1\}$.
    \end{enumerate}
    We will denote the set symbols satisfying these assumptions as $\setsymbols$.
\end{assumption}

Note that this implies $T(f)$ to be Hermitian and we have the symmetry property $\lambda_p(e^{\i\alpha}) = \lambda_p(e^{-\i\alpha})$.

We will now introduce two special kinds of Toeplitz matrices.
\begin{definition}[Banded matrices]
We call a Toeplitz operator $r$-banded if it has only $r$ nonzero off-diagonal blocks. Note that the symbol of an  $r$-banded Toeplitz operator is therefore a Laurent polynomial:
\begin{align*}
    f(z)=\sum_{\abs{s}\leq r} a_s z^s.
\end{align*}
\end{definition}

\begin{definition}[$r$-banded approximation]
For any symbol $f(z) = \sum_{s\in\Z}a_sz^s\in C_{k\times k}$ we can thus define the $r$-\emph{banded approximation} as 
\[
    f^{[r]}(z) \coloneqq \sum_{\abs{s}\leq r} a_s z^s 
\]
and call $T(f^{[r]})$ the $r$-\emph{banded approximation} of $T(f)$.
\end{definition}

We note that for $r$ large enough $f^{[r]}$ inherits \cref{assumptions} from $f$.

\begin{definition}[Circulant matrix]
    A $k$-circulant matrix of size $mk\times mk$ is a $k$-Toeplitz matrix such that its blocks satisfy
    $$
    a_{-j} = a_{m-j}, \quad j=1\dots m-1.
    $$
\end{definition}

A well-known result is that the eigenvectors of a circulant matrix are the Fourier modes.
\begin{lemma}\label{lemma: eigs of 1-circulant}
    Let $C_m(f)$ be a $1$-circulant matrix of size $m$ and $\lambda(\expi{\alpha})$ its only band. The eigenvalues of $C_m(f)$ are $\lambda(\expi{\alpha_j})$ and the eigenvectors are given by
    \begin{align*}
        \bm \omega_{j,m} \coloneqq \frac{1}{\sqrt{m}}(1,e^{\i\alpha_j},\dots,e^{\i\alpha_j(m-1)})^\top
    \end{align*}
    where $\alpha_j \coloneqq 2\pi\frac{j}{m}$ for $0\leq j\leq m-1$ and the superscript $\top$ denotes the transpose.
\end{lemma}

As we can see the spectrum of $C_m(f)$ is given by applying the band function to a uniform sampling of quasiperiodicities $\alpha_j\in [0,2\pi)$. 

To better capture the symmetry $\lambda(\expi{\alpha})=\lambda(e^{-\i\alpha})$ symmetry we can instead consider $\alpha_j\in [-\pi,\pi)$ by shifting the index by $\lfloor\frac{m}{2}\rfloor$ and define the following set of quasiperiodicity samples.
\begin{definition}[Discrete Brillouin zone]
    For $m\in \N$, we define the \emph{discretised Brillouin zone} as
    \begin{equation}
        Y^*_m \coloneqq\left\{2\pi\frac{j}{m}\,\middle|\, j=-\left\lfloor\frac{m}{2}\right\rfloor,\dots,m-1-\left\lfloor\frac{m}{2}\right\rfloor\right\}.
    \end{equation}
\end{definition}

Using this definition we can write $\sigma(C_m(f)) = \lambda(\expi{Y^*_m})$ and have the symmetry $\lambda(\expi{\alpha_j}) = \lambda(\expi{\alpha_{-j}})$. We further note that although $C_m(f)$ has coinciding eigenvalues, it remains diagonalisable, as the eigenvectors $\bm \omega_{j,m}$ form an orthonormal basis.

 To extend \cref{lemma: eigs of 1-circulant} to the block Toeplitz case, we define the quasiperiodic extension as follows.
\begin{definition}[Quasiperiodic extension]\label{def: quasiperiodic extension}
    Fix some $e^{\i\alpha}\in\mathbb{T}^1$ and consider a vector $\bm u \in \C^k$. Then, the $\alpha$-quasiperiodic extension of $\bm u$ of size $mk\in\N$ is given by
    \begin{align*}
    QP_m(\bm u,e^{\i\alpha})\coloneqq\frac{1}{\sqrt{m}}(\bm u, e^{\i\alpha}\bm u,\dots, e^{\i\alpha(m-1)}\bm u) \in \C^{mk}.
    \end{align*}
\end{definition}

\begin{proposition}\label{prop:eigs of k-circulant}
    Let $C_{mk}(f)$ be a $k$-circulant matrix with bands $\lambda_1,\dots \lambda_k$. For any $\alpha_j \in Y^*_m$ and $p\in \{ 1,\dots, k\}$, the circulant matrix has the eigenvalue $\lambda_p(\expi{\alpha_j})$ with corresponding eigenvector
    \[
        QP_m(\bm u_p(e^{\i\alpha_j}),e^{\i\alpha_j}).
    \]
    As in the $1$-circulant case, these eigenvectors form a complete orthonormal basis.
\end{proposition}

\subsection{Delocalisation}
\begin{definition}[Delocalisation]
Let $(\bm u_m)_m$ be a sequence of unit vectors in $\ell^2(\N)$, then $(\bm u_m)_m$ is said to be \emph{delocalised} if $\lim_{m\to\infty}\Vert \bm u_m\Vert_\infty =0$.
\end{definition}
Localisation and delocalisation of vectors in the finite regime is much more tricky and we will avoid to introduce a formal definition. Nevertheless, we will use the term \emph{localised} informally to indicate vectors $\bm v$ for which $\Vert \bm v \Vert_\infty / \Vert \bm v\Vert_2$ is ``large'' and \emph{delocalised} for those with $\Vert \bm v \Vert_\infty / \Vert \bm v\Vert_2$ close to $0$.

The \tftn can only work with delocalised eigenvectors: this is clear from the $e^{\i\alpha \cdot m}$ factor in \eqref{def:floquet-bloch-transform}. We hence seek to show that the matrices of our interest present plenty of delocalised eigenvectors. A first, rather simple result, shows that we have eigenvectors converging to $0$ weakly, but this will not be sufficient.
\begin{lemma}\label{prop: weak conv of eigenvectors to essential spectrum}
    Let $T, T_m$ be self-adjoint, bounded linear operators on a Hilbert space $H$. Let $\lambda\in \sigma_{ess}(T)$ and $\lambda_m,\bm u_m$ be an eigenpair of $T_m$ with $\lambda_m\to\lambda$ and $\norm{\bm u_m}=1$. If $T_m\to T$ strongly,
    then $\bm u_m$ has a subsequence $\bm u_n$ converging weakly to 0.
\end{lemma}
\begin{proof}
    Because $\norm{\bm u_m}=1$ is bounded, we use the Banach-Alaoglu theorem (see \cite[Theorem 3.18]{brezisFunctionalAnalysisSobolev2011}) to find a weakly convergent subsequence $\bm u_n\to \bm u\in H$. From \cite[Proposition 3.13]{brezisFunctionalAnalysisSobolev2011}, we know that this implies 
    \[
    \ip{\bm u_n}{\bm x_n}\to \ip{\bm u}{\bm x} \quad \text{for all }\bm x_n\to \bm x \text{ strongly.}
    \]

    We now aim to show that $\bm u=0$. For any $\bm x\in H$, we have
    \[
    \ip{(T-\lambda_n)\bm u_n}{\bm x} = \ip{(T-T_n)\bm u_n}{\bm x} = \ip{\bm u_n}{(T-T_n)\bm x} \to 0,
    \] since $(T-T_n)\bm x\to 0$ strongly.
    At the same time, we have 
    \[
    \ip{(T-\lambda_n)\bm u_n}{\bm x} = \ip{\bm u_n}{(T-\lambda)\bm x+(\lambda-\lambda_n)\bm x} \to \ip{\bm u}{(T-\lambda)\bm x} = \ip{(T-\lambda)\bm u}{\bm x},
    \]
    because $(\lambda-\lambda_n)\bm x\to 0$ strongly.
    Putting the two equalities together, this implies that $\ip{(T-\lambda)\bm u}{\bm x} = 0$ for all $\bm x\in H$. Assuming $\bm u\neq 0$  would imply that $\bm u\in H$ is an eigenvector of $T$, a contradiction because $\lambda$ lies in the essential spectrum.
\end{proof}
\begin{remark}\label{rmk:not yet delocalized}
    While we would like to prove delocalisation, \emph{i.e.}, $\norm{\bm u_n}_\infty\to 0$, \cref{prop: weak conv of eigenvectors to essential spectrum} only ensures the weaker version: For all $j\in \N$, we have $\abs{\bm u_n^{(j)}}\to 0$ as $n\to\infty$. Unfortunately, this does not rule out \enquote{spiky} eigenvectors, where the spike keeps moving to the right as $n$ increases. It is well known, for example, that the sequences $(\bm e_n)_n$, where $e_n^{(j)}=\delta_{n,j}$, converges weakly but not strongly to $0$.
\end{remark}

To fix the issue pointed out in \cref{rmk:not yet delocalized} we first need an intermediate result. We use the following notation for the standard inclusion map $\ell^2(\N) \to \ell^2(\Z)$
\begin{align*}
    J: \ell^2(\N) &\to \ell^2(\Z)\\
    J(x)_j &= \begin{dcases}
        x_j, & j\geq 0,\\
        0, & j< 0,
    \end{dcases}
\end{align*}
and the projection map $\ell^2(\Z) \to \ell^2(\N)$
\begin{align*}
    P: \ell^2(\Z) &\to \ell^2(\N)\\
    (x_j)_{j\in\Z} &\mapsto (x_j)_{j\in\N}.
\end{align*}
\begin{lemma}\label{lem:strongly conv shift}
    Consider a symbol $f\in C_{k\times k}$ and denote by $s_q$ the unitary left-shift operator acting on $\ell^2(\Z)$ by shifting a sequence by $q$ entries to the left, that is,
    \begin{align}
        \label{eq: sm op}
        s_q:(x_j)_{j\in \Z}\mapsto (x_{j+q})_{j\in \Z}.
    \end{align} 
    Let $(\beta_m)_m\subset \Z$ be a sequence of indices. Then,  $s_{\beta_m}JT_{mk}(f)Ps_{\beta_m}^*$ is again a bounded, linear, self-adjoint operator and we have
    \[
        s_{\beta_m}JT_m(f)Ps_{\beta_m}^* \to L(f) \text{ strongly as } m\to\infty
    \]
    if and only if $\beta_m \in \{0,k,\dots mk\}$ and infinitely separated from $0$ and $mk$, \emph{i.e.}, $\abs{\beta_m} \to \infty$ and $\abs{mk-\beta_m}\to \infty$ as $m\to \infty$.
\end{lemma}
\begin{proof}
    The assumption of $\beta_m$ being between $0$ and $m$ and infinitely separated ensures that $s_{\beta_m}JT_m(f)Ps_{\beta_m}^*$ accumulates arbitrarily many nonzero entries to both sides of the origin. Where they are nonzero, these entries match the entries of the Laurent operator because the $k$-Laurent operator is invariant under shifts by multiples of $k$. We can then repeat the arguments of \cref{lemma: norm of toeplitz sections converges} and \cref{lemma: strong convergence of sections}.
\end{proof}
We can now show delocalisation for the eigenvectors of Toeplitz matrices.
\begin{proposition}\label{prop:final result about delocalisation}
    Consider a symbol $f\in C_{k\times k}$. Let $\lambda\in \sigma_{ess}(T(f))$ and $(\lambda_m,\bm u_m)$ eigenpairs of $T_{mk}(f)$ with $\lambda_m\to\lambda$ and $\norm{\bm u_m}=1$. 
    Then, there exists a subsequence $\bm u_n$ converging strongly to $0$ with respect to $\norm{\cdot}_\infty$.
\end{proposition}
\begin{proof}
    For now we consider the non-block case, \emph{i.e.}, $k=1$.
    We denote $\beta_m=\argmax_{m\in \Z} \abs{\bm u_m}$ and note that $\beta_m\in \{0,\dots,m\}$ because $u_m$ is an eigenvector of $T_m(f)$.
    After possibly taking subsequences, we can now distinguish three cases: Either $\abs{\beta_m}$ is bounded, $\abs{m-\beta_m}$ is bounded or both are unbounded.
    The first case is handled by \cref{rmk:not yet delocalized} and the second one can be handled similarly.
    We shall now focus on the last case and embed $T_m(f)$ and $\bm u_m$ into $\ell^2(\Z)$ with $J$. Then we can  apply \cref{lem:strongly conv shift} to find that $s_{\beta_m}JT_m(f)P s_{\beta_m}^* \to L(f)$ strongly. $s_{\beta_m}\bm u_m$ is now a sequence of eigenvectors of $s_{\beta_m}JT_m(f)P s_{\beta_m}^*$ with $\lambda\in \sigma_{ess}(T(f)) = \sigma_{ess}(L(f))$ see \cite{gohberg.kaashoek.ea1993Classes}. We apply \cref{prop: weak conv of eigenvectors to essential spectrum} to find a subsequence $s_{\beta_n}\bm u_n$ converging weakly to $0$. Finally, because $\beta_n$ is the index of the maximal entry of $\bm u_n$, we have that 
    \[
    \norm{\bm u_n}_\infty = \norm{s_{\beta_n}\bm u_n}_\infty = (s_{\beta_n}\bm u_n)^{(0)} = \ip{s_{\beta_n}\bm u_n}{\bm e_0} \to 0,
    \]
    where we have used the fact that $s_{\beta_n}\bm u_n\to 0$ weakly.
    
    The proof in the block case, $k> 1$, works analogously by choosing $\beta_m$ to be the first index of the block of $\bm u_m$ containing its maximal entry, \emph{i.e.}, $\beta_m = \lfloor\argmax_{m\in \Z} \abs{\bm u_m}/k\rfloor k$. We then get $\ip{s_{\beta_n}\bm u_n}{\bm e_p} \to 0$ for $p=0,\dots k-1$, yielding the desired result.
\end{proof}

While \cref{prop:final result about delocalisation} only proves the existence of a delocalised subsequence we can immediately use it to prove that any sequence of eigenvectors with eigenvalues converging to the essential spectrum must be delocalised.
\begin{corollary}\label{cor:general delocalization}
    Consider a symbol $f\in C_{k\times k}$. Let $\lambda\in \sigma_{ess}(T(f))$ and a sequence $(\lambda_m,\bm u_m)$ of eigenpairs of $T_{mk}(f)$ with $\lambda_m\to\lambda$ and $\norm{\bm u_m}=1$. 
    Then $\bm u_m$ must be delocalised.
\end{corollary}
\begin{proof}
    Suppose by contradiction that for some $\epsilon>0$, $\bm u_m$ is such that $\norm{\bm u_m}_\infty>\varepsilon$ for all $m\in \N$. Then, by \cref{prop:final result about delocalisation}, we find a subsequence $\bm u_n$ of $\bm u_m$ with $\norm{\bm u_n}_\infty\to 0$, a contradiction. 
\end{proof}
For large finite Toeplitz matrices, the bulk of eigenmodes is thus delocalised, allowing us to reconstruct quasiperiodicity information.

\section{Truncated Floquet-Bloch transform}\label{sec: DFT}
The key tool that we will use is the discrete Fourier transform.
\begin{definition}[Discrete Fourier transform]
    We define the \emph{discrete Fourier transform} as
    \begin{equation}
    \begin{aligned}
        \dfrt:\C^m &\to \C^m\\
        \bm u &\mapsto \left(\frac{1}{\sqrt{m}}\sum_{s=0}^{m-1}\bm u^{(s)}e^{-\i2\pi\frac{j}{m}s}\right)^{(j)}\quad 0\leq j\leq m-1.
    \end{aligned}
    \end{equation}
    We recall the shorthand $\alpha_j \coloneqq 2\pi\frac{j}{m}$ and $\bm \omega_{j,m} \coloneqq \frac{1}{\sqrt{m}}(1,e^{\i\alpha_j},\dots,e^{\i\alpha_j(m-1)})^\top$ and find the following matrix representation of $\dfrt$:
    \begin{equation}
        \dfrt = \begin{pmatrix}
            | & & |\\
            \bm \omega_{0,m}& \dots & \bm \omega_{m-1,m}\\
            | && |
        \end{pmatrix} = \begin{pmatrix}
            \text{---} & \bm \omega_{0,m}^\top& \text{---}\\ &\dots& \\ \text{---}& \bm \omega_{m-1,m}^\top& \text{---}
        \end{pmatrix}.
    \end{equation}
    We also recall that the discrete Fourier transform is a unitary isomorphism.
\end{definition}

When considering block Toeplitz matrices $T_{mk}(f)\in \C^{mk\times mk}$, it is often useful to utilise the following unitary decomposition
\begin{equation}
\begin{aligned}
\label{eq: phi map slicing}
    \Phi: \C^{mk} &\to (\C^m)^k = \overbrace{\C^m\oplus\dots \oplus \C^m}^{k-\text{times}}\\
    \bm u &\mapsto (\Phi_1(\bm u),\dots,\Phi_k(\bm u)),
\end{aligned}
\end{equation}
where $\Phi_p(\bm u)\in \C^m$ denotes the \emph{$p$\textsuperscript{th} section} of $\bm u$, given by $(\Phi_p(\bm u))^{(s)} \coloneqq \bm u^{(p+sk)}$ for $s=0,\dots,m-1$. Here, we equip $\C^m\oplus\dots \oplus \C^m$ with the canonical norm $\norm{u} = \sqrt{\sum_{p=1}^k \norm{\Phi_p(u)}^2}$.

We can now define the truncated Floquet-Bloch transform, which is just a particular packing of the discrete Fourier transform. 

\begin{definition}[Truncated Floquet-Bloch transform]
\label{def: tfbt}
    We now define the \emph{truncated Floquet-Bloch transform} as the section-wise discrete Fourier transform
    \begin{equation}
    \begin{aligned}
        \mc{T}: (\C^m)^k &\to (\C^m)^k\\
        (\Phi_1(\bm u),\dots, \Phi_k(\bm u)) &\mapsto (\dfrt[\Phi_1(\bm u)],\dots,\dfrt[\Phi_k(\bm u)]),
    \end{aligned}
    \end{equation}
    which is again a unitary isomorphism. For vectors $\bm u\in \C^{mk}$, we will often write $\mc T(\bm u)$ to mean $\mc T(\Phi(\bm u))$ where it is clear from the context.
\end{definition}

The truncated Floquet-Bloch transform defined in \cref{def: tfbt} has, at first sight, little to do with the one originally presented in \cite{ammari.davies.ea2023Convergence, ammari.davies.ea2023Spectral}. However, they are very much related as \cref{fig: old-new-tfbt} shows. The formulation of \cref{def: tfbt} has a couple of advantages. The first is computational, although this is only minor and only occurs for very large structures. The second one is the mathematical framework, which thanks to the vast literature on the discrete Fourier transform, will allow us to prove various theorems.

\begin{figure}[h]
    \centering
    \includegraphics[width=0.5\linewidth]{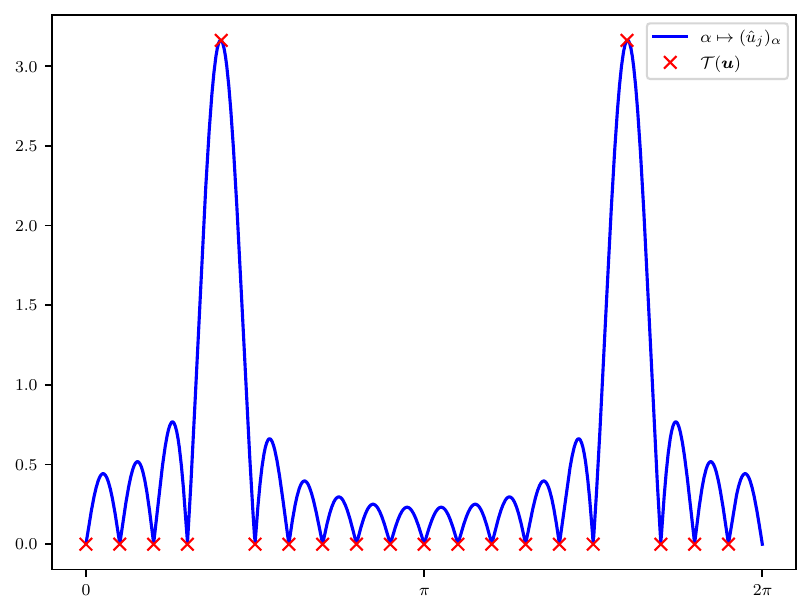}
    \caption{The ``original'' truncated Floquet-Bloch transform presented in \cite{ammari.davies.ea2023Convergence, ammari.davies.ea2023Spectral} in blue and the version from \cref{def: tfbt} with indices scaled as $r\mapsto 2\pi\frac{r}{m}$ for $0\leq r\leq m$ in red. The plot shows the two techniques applied to an eigenvector of a circulant matrix of size $20\times 20$.}
    \label{fig: old-new-tfbt}
\end{figure}

Intuitively, the $j$-th \enquote{row} of $\mc{T}(\bm u)$ contains the Fourier coefficients with respect to $\alpha_j\in Y^*_m$, obtained by taking the discrete Fourier transform of $\bm u$ section-wise for each of the $k$ sections.
To extract quasiperiodicity information from truncated Floquet-Bloch transform, it is thus useful to denote the rows as follows.
\begin{definition}[Truncated Floquet-Bloch projection]\label{def:dpf}
    We define the $j$\textsuperscript{th} \emph{truncated Floquet-Bloch projection} for $1\leq j\leq m$ as 
    \begin{equation}
        \begin{aligned}
            \mc{T}^j:\C^{mk} &\to \C^k\\
            \bm u &\mapsto ((\mathcal{F}(\Phi_1\bm u))^{(j)},\dots, (\mathcal{F}(\Phi_k\bm u))^{(j)})^\top.
        \end{aligned}
    \end{equation}
    For some $\bm u \in \C^{mk}$ this corresponds to calculating the \tftn of $\mc T(\bm u) \in (\C^m)^k$ and taking the $j$\textsuperscript{th} \enquote{row} of the result.
\end{definition}
A good way of picturing what $\mathcal{T}^j$ does is to look at \cref{fig: old-new-tfbt} where we plotted in red $(\mathcal{T}^j(\bm u))_{0\leq j\leq m-1}$. What we are actually interested in is the index of the biggest Fourier coefficient, as this one is the one that leads the periodic behaviour of $\bm u$. The following definition is a better-behaved version of taking the $\argmax$ of $(\mathcal{T}^j(\bm u))_{0\leq j\leq m-1}$ to identify the peak quasiperiodicity.
\begin{definition}[Discrete quasiperiodicity]
    Let $\bm u \in \C^{mk}$. Then the \emph{discrete quasiperiodicity} associated to $\bm u$ is given by
    \begin{align*}
        \mathcal{Q}_m(\bm u)\coloneqq \sum_{\alpha_j\in Y^*_m} \abs{\alpha_j}\norm{\mathcal{T}^{j}(\bm u)}^2, 
    \end{align*}
    where we extend $\mc T^j$ to negative indices $j$ by taking $j$ modulo $m$.
\end{definition}
This amounts to taking the weighted average of quasiperiodicities, using $\norm{\mathcal{T}^{j}(\bm u)}^2$ to measure resonance strength. We take the absolute value $\abs{\alpha_j}$ to account for the symmetry $\lambda_p(\expi{\alpha})=\lambda_p(e^{-\i\alpha})$.

\subsection{Truncated Floquet-Bloch transform and circulant matrices}
To elucidate the \tftn we investigate how it interacts with the eigenvectors of circulant matrixes $C_{mk}(f)$.

\begin{example}\label{ex:tfbt for circulant eves}
     By \cref{prop:eigs of k-circulant}, the eigenvectors of $C_{mk}(f)$ are given by the quasiperiodic extension of symbol matrix eigenvectors: $QP_m(\bm u_p(\expi{\alpha_s}),\expi{\alpha_s})$. Their section decomposition is given by 
    \[
    \Phi(QP_m(\bm u_p(\expi{\alpha_s}),\expi{\alpha_s})) = (\bm u_{p}^{(1)}(\expi{\alpha_s})\bm \omega_{s,m}, \dots,  \bm u_{p}^{(k)}(\expi{\alpha_s})\bm \omega_{s,m}).
    \]
    Using $\dfrt(\bm \omega_{s,m}) = \bm e_s$, we then find that 
    \[
        \mc{T}[QP_m(\bm u_p(\expi{\alpha_s}),\expi{\alpha_s})] = (\bm u_{p}^{(1)}(\expi{\alpha_s}) \bm e_s, \dots,  \bm u_{p}^{(k)}(\expi{\alpha_s})\bm e_s)
    \] and 
    \[
        \mc{T}^j[QP_m(\bm u_p(\expi{\alpha_s}),\expi{\alpha_s})] = (\bm u_{p}^{(1)}(\expi{\alpha_s}) \delta_{js}, \dots,  \bm u_{p}^{(k)}(\expi{\alpha_s})\delta_{js})^\perp.
    \]
    Intuitively, because the circulant eigenvectors are the $\alpha_s$-quasiperiodic extensions of the symbol eigenvectors, each of their sections must be proportional to $\bm \omega_{s,m}$. Applying the discrete Fourier transform section-wise, we find that for every section, the only nonzero Fourier coefficient is at entry $s$. 
    Because $\norm{\bm u_p(\expi{\alpha_s})}=1$ this yields $\norm{\mc{T}^j[QP_m(\bm u_p(\expi{\alpha_s}),\expi{\alpha_s})]}^2 = \delta_{js}$, and we can find that
    \begin{align*}
        \mathcal{Q}_m(QP_m(\bm u_p(\expi{\alpha_s}),\expi{\alpha_s})))=\sum_{\alpha_j\in Y^*_m} \abs{\alpha_j} \delta_{js} =  \alpha_s.
    \end{align*}

    For eigenvectors of $C_{mk}(f)$ we can thus recover the quasiperiodicity using the Floquet-Bloch transform. 
    
    This reconstruction also respects the symmetry $\lambda_p(\expi{\alpha})=\lambda_p(e^{-\i\alpha})$ as follows
    \begin{gather*}
        \mathcal{Q}_m(c_1QP_m(\bm u_p(\expi{\alpha_s}),\expi{\alpha_s}))+c_2QP_m(\bm u_p(\expi{\alpha_{-s}}),\expi{\alpha_{-s}}))) = \sum_{\alpha_j\in Y^*_m} \abs{\alpha_j} (\abs{c_1}^2\delta_{js}+\abs{c_2}^2\delta_{j,-s})\\
        =\abs{c_1}^2\abs{\alpha_s}+\abs{c_1}^2\abs{\alpha_{-s}} = (\abs{c_1}^2+\abs{c_2}^2)\alpha_s,
    \end{gather*}
    where we assumed $s>0$ without loss of generality. For any $\bm u$ in the eigenspace of $\lambda_p(\expi{\alpha_s})$, $\mc Q_m(\bm u)$ thus returns the the quasiperiodicity $\abs{\alpha_s}$.
\end{example}

Having observed that the truncated Floquet-Bloch transform recovers the quasiperiodicity of circulant matrix eigenvectors, we now aim to extend our understanding to pseudo-eigenvectors. This will then enable us to extend the quasiperiodicity recovery to more complex systems.

\begin{definition}
    For $A\in \C^{n\times n}$ a normal matrix and $\lambda_\varepsilon\in \C$,
    we can decompose $\C^n$ into the \emph{near and far eigenspaces}:
    \[E_\varepsilon \coloneqq \bigoplus_{\substack{\lambda'\in \sigma(A)\\\abs{\lambda'-\lambda_\varepsilon}\leq \varepsilon}}\operatorname{Eig}_{\lambda'}(A)\quad\text{ and }\quad E^\perp_\varepsilon = \bigoplus_{\substack{\lambda'\in \sigma(A)\\\abs{\lambda'-\lambda_\varepsilon}>\varepsilon}}\operatorname{Eig}_{\lambda'}(A)\] with decomposition $\C^n = E_\varepsilon \oplus E^\perp_\varepsilon$.
\end{definition}
For a normal matrix $A$, if $(\lambda, \bm u)$ is a pseudo-eigenpair of $A$, then $\bm u$ is well-approximated by vectors laying in the near eigenspace of $\lambda$. The following result makes this formal.
\begin{lemma}\label{lem:eigenspaceconv}
    Let $\varepsilon>0$. Let $A\in \C^{n\times n}$ be a normal matrix and let $(\lambda_\varepsilon,u_\varepsilon)$ be an $\varepsilon^2$-pseudo-eigenpair, \emph{i.e.}, $\norm{(A-\lambda_\varepsilon)u_\varepsilon}<\varepsilon^2$ and $\norm{u_\varepsilon}=1$. Then, decomposing $u_\varepsilon = u_\parallel + u_\perp \in E_\varepsilon\oplus E_\varepsilon^\perp$,  we find that
    \[
        \norm{u_\perp} < \varepsilon \quad \text{ and } \quad \norm{u_\parallel} > \sqrt{1 - \varepsilon^2}.
    \]
\end{lemma}
\begin{proof}
    We assume without loss of generality that $A$ has no double eigenvalues. We diagonalise $A = UDU^*$ with $UU^*=\operatorname{Id}$ and $\operatorname{Id}$ being the identity. Then, we write $\Tilde{u}_\varepsilon\coloneqq U^*u_\varepsilon$ and calculate
    \begin{gather*}
        \varepsilon^4>\norm{Au_\varepsilon-\lambda_\varepsilon I_n u_\varepsilon}^2 
        = \norm{UDU^*u_\varepsilon-\lambda_\varepsilon UI_nU^* u_\varepsilon}^2 
        = \norm{D\Tilde{u}_\varepsilon-\lambda_\varepsilon I_n\Tilde{u}_\varepsilon}^2\\ 
        = \norm{(D-\lambda_\varepsilon I_n)\Tilde{u}_\varepsilon}^2
        = \sum_{j=1}^n\abs{\lambda_i-\lambda_\varepsilon}^2 \abs{\Tilde{u}_\varepsilon^{(j)}}^2.
    \end{gather*}
    Denoting by $u_j$ the unit eigenvector associated to $\lambda_j$, we find that $$\abs{\Tilde{u}_\varepsilon^{(j)}}^2 = \abs{\ip{\Tilde{u}_\varepsilon}{e_j}}^2 = \abs{\ip{U\Tilde{u}_\varepsilon}{Ue_j}}^2 = \abs{\ip{u_\varepsilon}{u_j}}^2.$$ If we now restrict the above sum to $E_\varepsilon^\perp$, then  we obtain that
    \begin{gather*}
        \varepsilon^4>  \sum_{(\lambda',u')\in E_\varepsilon^\perp}\abs{\lambda'-\lambda_\varepsilon}^2 \abs{\ip{u_\varepsilon}{u'}}^2 > \sum_{u'\in E_\varepsilon^\perp}\varepsilon^2 \abs{\ip{u_\varepsilon}{u'}}^2,
      \end{gather*}
      and thus
      \begin{gather*}
        \varepsilon > \sqrt{\sum_{u'\in E_\varepsilon^\perp} \abs{\ip{u_\varepsilon}{u'}}^2} = \norm{u_\perp}.
    \end{gather*}
    Finally, orthogonality yields
    \[
    1 = \norm{u_\varepsilon}^2 = \norm{u_\parallel}^2 + \underbrace{\norm{u_\perp}^2}_{<\varepsilon^2}
    \]
    and thus $\norm{u_\parallel} > \sqrt{1 - \varepsilon^2}$.
\end{proof}
\begin{lemma}\label{lem:bandinverse}
    Let $f$ be a continuously differentiable symbol and let $\lambda_0 =  \lambda_p(e^{\i\alpha_0})$ where $e^{\i\alpha_0}\neq \pm1$.
    Then for $\varepsilon$ small enough
    \begin{align}
         (\lambda_0-\epsilon,\lambda_0+\epsilon)\cap \sigma(C_{mk}(f))\subset \lambda_p(\expi{Y^*_{m,\kappa\varepsilon}(\alpha_0)})
    \end{align}
    where $\kappa = 2
    \abs{\frac{1}{(\lambda_p)'(\expi{\alpha_0})}}$ and 
    \[
        Y^*_{m,\kappa\varepsilon}(\alpha_0) = \big((\alpha_0-\kappa\varepsilon, \alpha_0+\kappa\varepsilon) \cup (-\alpha_0-\kappa\varepsilon, -\alpha_0+\kappa\varepsilon)\big) \cap Y^*_m .
    \]
\end{lemma}
\begin{proof}
    For $\varepsilon>0$ small enough, the entire $\varepsilon$-neighbourhood of $\lambda_0$ is contained in the band, as $\lambda_0$ lies its interior, \emph{i.e.}, $[\lambda_0-\varepsilon,\lambda_0+\varepsilon] \subset \operatorname{int} \lambda_p(\TO)$.
    Since the bands do not cross each other, this yields  $[\lambda_0-\varepsilon,\lambda_0+\varepsilon]\cap \lambda_{q}(\mathbb{T})\neq \emptyset$ if and only if $q =  p$. By the symmetry $\lambda_p(\expi{\alpha}) = \lambda_p(e^{-\i\alpha})$ of the band functions we can restrict ourselves to $\TO\cap \mc{H}$ for now, where $\mc{H}\coloneqq\{z\in \C\mid \Im z>0\}$ denotes the upper half of the complex plane. Because we have assumed that $\lambda_p$ has no van der Hove points, it must be injective on $\TO\cap \mc{H}$ and we can find the continuous inverse $\lambda_p^{-1}:[\lambda_0-\varepsilon,\lambda_0+\varepsilon] \to \TO\cap \mc{H}$. Since $\lambda_p^{-1}$ is defined on a compact set, it is Lipschitz continuous with constant $C_1 = \max_{\lambda'\in [\lambda_0-\varepsilon,\lambda_0+\varepsilon]}
    \abs{\frac{1}{(\lambda_p)'(\lambda^{-1}_p(\lambda'))}}$. Furthermore, for $\varepsilon$ small enough, we can use the fact that $(\lambda_p)'(\lambda^{-1}_p(\lambda'))$ is continuous around $\alpha_0 = \lambda^{-1}_p(\lambda_0)$ to instead choose the Lipschitz constant $\kappa = 2\abs{\frac{1}{(\lambda_p)'(\expi{\alpha_0})}} \geq C_1$, which is well-defined as we assumed no van der Hove points.
    
    We can use the Lipschitz continuity of $\lambda_p^{-1}$ to bound 
    \[
    (\lambda_0-\varepsilon,\lambda_0+\varepsilon) \subset \lambda_p\left(\left\{\expi{\alpha}:\alpha \in (\alpha_0-\kappa\varepsilon, \alpha_0+\kappa\varepsilon) \cup (-\alpha_0-\kappa\varepsilon, -\alpha_0+\kappa\varepsilon) \right\} \right),
    \]
    where we incurred $\pm\alpha_0$ because of the band function symmetry. Finally, we  use $\sigma(C_{mk}(f)) = \cup_{p'} \lambda_{p'}(\{\expi{\alpha_{j}}\mid \alpha_j\in Y^*_m\})$ to get the desired inclusion.
\end{proof}

\begin{definition}[Convergent pseudo-eigenpair]
    Let $T_m$ be a sequence of normal operators converging strongly to some normal operator $T$. We say that a sequence of tuples $(\lambda_m,\bm u_m)$ with $\norm{u_m} = 1$ is a \emph{convergent sequence of pseudo-eigenpairs} to $\lambda$ if $\lambda_m\to \lambda$ and $\norm{T_m\bm u_m - \lambda_m\bm u_m}\to 0$ as $m\to \infty$.
\end{definition}

Note that the fact that $T_m$ and $T$ are normal implies that $\bm u_m$ is a Weyl sequence for $T$ at $\lambda$ and thus $\lambda\in \sigma_{ess}(T)$.

The following result shows that analogously to \cref{ex:tfbt for circulant eves}, the \tftn converges to an approximate Kronecker delta vector also for convergent pseudo-eigenvectors of circulant matrices.
\begin{proposition}\label{prop: circulant conv pEve has spiky TFT}
    Consider a symbol $f\in C_{k\times k}$ and let $(\lambda_m,\bm u_m)$ be a convergent sequence of pseudo-eigenpairs of $C_{mk}(f)$ where $\lambda_m\to\lambda_0=\lambda_p(\expi{\alpha_0})$ such that $e^{\i\alpha_0}\neq \pm1$. Let $\varepsilon>0$ be sufficiently small and $m$ large enough such that $\norm{C_{mk}(f)u_m-\lambda_mu_m}<\varepsilon^2$ and $\abs{\lambda_m-\lambda_0}<\varepsilon$, then we have
    \[
        \sum_{\alpha_j\in(Y^*_{m,\delta}(\alpha_0))^c}\norm{\mc{T}^j(\bm u_m)}^2 < \varepsilon^2 \quad \text{and} \quad \sum_{\alpha_j\in Y^*_{m,\delta}(\alpha_0)}\norm{\mc{T}^j(\bm u_m)}^2 > 1 - \varepsilon^2 , 
    \]
    where \begin{equation} \label{delta} \delta = 
    \frac{4\varepsilon}{\abs{(\lambda_p)'(\expi{\alpha_0})}},\end{equation}
    $Y^*_{m,\delta}$ as in \cref{lem:bandinverse} and $(Y^*_{m,\delta}(\alpha_0))^c = Y^*_m\setminus Y^*_{m,\delta}(\alpha_0)$.
\end{proposition}
\begin{proof}
    Using \cref{lem:eigenspaceconv}, we can decompose $\bm u_m = \bm u_\parallel \in E_\delta + \bm u_\perp \in E_\delta^\perp$. Because $e^{\i\alpha_0}\neq \pm1$, $\lambda_0$ must be away from the band edge and inside the $p$\textsuperscript{th} band. \cref{lem:bandinverse} then applies for $\varepsilon$ small enough to obtain
    \[
         (\lambda_m-\varepsilon,\lambda_m+\varepsilon)\cap \sigma(C_{mk}(f))\subset (\lambda_0-2\varepsilon,\lambda_0+2\varepsilon)\cap \sigma(C_{mk}(f)) \subset \lambda_p(e^{\i Y^*_{m,\delta}(\alpha_0)}),
    \]
    where $\delta$ is given by \eqref{delta} and $Y^*_\delta$ as in \cref{lem:bandinverse}.
    Consequently, we have
    \[
    \bm u_\parallel\in E_\varepsilon \subset \bigoplus_{\substack{\lambda_q=\expi{\alpha_s}\\\alpha_s\in Y^*_{m,\delta}(\alpha_0)}} \operatorname{Eig}_{\lambda_q}(C_{mk}(f)) \eqqcolon F_\varepsilon,
    \]
    which allows us to orthogonally decompose $\bm u_m = \bm v_\parallel\in F_\varepsilon + \bm v_\perp\in F_\varepsilon^\perp$
    and by $E_\varepsilon \subset F_\varepsilon$, we have $\norm{\bm v_\parallel} \geq \norm{\bm u_\parallel}$ and $\norm{\bm v_\perp} \leq \norm{\bm u_\perp}$.

    Now we want to understand $\mc T(\bm u_m)$ and first aim to show that $\mc T^j(\bm u_m) = \mc T^j(\bm v_\perp)$ if $\alpha_j\notin  Y^*_{m,\delta}(\alpha_0)$, which by the linearity of $\mc T^j$ is equivalent to proving $\mc T^j(\bm v_\parallel)=\bm 0$ for such $j$.
    We can write  
    \[
        \bm v_\parallel = \sum_{\alpha_s\in Y^*_{m,\delta}(\alpha_0)} c_s QP_m(\bm u_p(\expi{\alpha_s}),\expi{\alpha_s}),
    \]
    but by \cref{ex:tfbt for circulant eves} we know that $\norm{\mc{T}^j[QP_m(\bm u_p(\expi{\alpha_s}),\expi{\alpha_s})]}^2 = \delta_{js}$ and 
    \[
    \norm{\mc T^j(\bm v_\parallel)}^2 = \sum_{\alpha_s\in Y^*_{m,\delta}(\alpha_0)} \abs{c_s}^2 \delta_{js},
    \]
    equals zero for $\alpha_j\notin Y^*_{m,\delta}(\alpha_0)$, as desired.

    Finally, we calculate 
    \[
    \sum_{\alpha_j \in (Y^*_{m,\delta}(\alpha_0))^c}\norm{\mc{T}^j(\bm u_m)}^2 = \sum_{\alpha_j \in (Y^*_{m,\delta}(\alpha_0))^c}\norm{\mc{T}^j(\bm v_\perp)}^2 \leq \sum_{\alpha_j \in Y^*_{m}}\norm{\mc{T}^j(\bm v_\perp)}^2 = \norm{\mc{T}(\bm v_\perp)}^2 < \varepsilon^2
    \]
    and 
    \[
    \sum_{\alpha_j\in Y^*_{m,\delta}(\alpha_0)}\norm{\mc{T}^j(\bm u_m)}^2 = 1-\sum_{\alpha_j\in (Y^*_{m,\delta}(\alpha_0))^c}\norm{\mc{T}^j(\bm u_m)}^2 > 1-\varepsilon^2.
    \]
\end{proof}
\begin{corollary}\label{cor: Q recovers quasifrequency in circulant}
    Consider $f\in C_{k\times k}$ with $C_{mk}(f)$ the corresponding sequence of circulant matrices. Let $(\lambda_m,\bm u_m)$ be a convergent sequence of pseudo-eigenpairs to $\lambda_0=\lambda_p(\expi{\alpha_0})$ such that $e^{\i\alpha_0}\neq \pm1$. Then,
    $$
    \alpha_0 = \pm \lim_{m\to\infty}\mathcal{Q}_m(\bm u_{m}).
    $$
\end{corollary}

\section{Periodic structures} \label{sec:5}
In this section, we will apply the newly defined truncated Floquet-Bloch transform from \cref{sec: DFT} to spatially periodic physical systems and show that as the size of the system grows the truncated Floquet-Bloch transform allows us to recover the band structure accurately and efficiently. We first illustrate a simpler system (a single resonator with only nearest neighbour interactions) and then provide a full generalisation of our result.

It has been shown \cite[Theorem 4.1.]{ammari.davies.ea2023Spectral} that any eigenmode of the infinite structure has a corresponding approximation in its large, finite counterpart. Furthermore, the bulk of eigenmodes of a finite system are such approximations of the infinite modes.
\begin{proposition}\label{prop: erik bryn spectral converngence}
    Let $e^{\i\alpha}\in \mathbb{T}^1$ and $f\in\setsymbols$. Consider an eigenpair $(\lambda(e^{\i\alpha}),\bm u(e^{\i\alpha}))$ of $f(e^{\i\alpha})$. Then, there exits a sequence $(m,j_m)\subset \N^2$ and eigenpairs $T_{m}(f)\bm u_{(m,j_m)} = \lambda_{(m,j_m)} \bm u_{(m,j_m)}$ so that
    \begin{align}
    \label{eq: convergence of eva and eve to quasiperiodic extension eve}
        \lambda_{(m,j_m)} \to \lambda(e^{\i\alpha}) \quad \text{and}\quad \Vert \bm u_{(m,j_m)} - QP_m(\bm u(e^{\i\alpha}),\alpha)\Vert \to 0 \quad\text{as } m\to\infty.
    \end{align}
    Here, $QP_m(\bm u(e^{\i\alpha}),e^{\i\alpha})$ denotes the quasiperiodic extension from \cref{def: quasiperiodic extension}. 
\end{proposition}

We aim to show that \cref{prop: erik bryn spectral converngence} can be reversed through the truncated Floquet-Bloch transform by recovering the limiting quasiperiodicity from a finite eigenmode.

\subsection{Single resonator or particle with only nearest neighbour interactions}\label{subsec: 1D periodic monomers}
Structures obtained by repeating periodically a single resonator or a particle keeping only nearest neighbour interactions are modelled with a Toeplitz matrix of the following form: 

\begin{align}\label{eq: form cap mat 1D}
    \mathcal{C}_m = \begin{pmatrix}
        a_0 + a_{-1}   & a_1   & 0     & 0     & \cdots & 0     & 0 \\
a_{-1} & a_0   & a_1   & 0     & \cdots & 0     & 0 \\
0      & a_{-1} & a_0   & a_1   & \cdots & 0     & 0 \\
0      & 0     & \ddots & \ddots   & \ddots & 0     & 0 \\
\vdots & \vdots & \ddots & \ddots & \ddots & \ddots & \vdots \\
0      & 0     & 0     & 0     & \ddots & a_0   & a_1 \\
0      & 0     & 0     & 0     & \cdots & a_{-1} & a_0 + a_{1} \\
    \end{pmatrix} \in \R^{m\times m} .
\end{align}

 This is the case in one-dimensional systems of resonators \cite{feppon.ammari2022Subwavelength, ammari.barandun.ea2023Edge} but also a common approximation in condensed matter physics; see, for instance, \cite{cbms2}.

The following result holds. 
\begin{proposition}\label{prop: DFT gives two diracts on tridiagonal toeplitz and 1D cap mat}
    Let $M_m\in\C^{m\times m}$ be any of $T_m(f)$ for $f(z)=a_{-1} z^{-1}+a_0 + a_1z^1$ or $\mathcal{C}_m$ from \eqref{eq: form cap mat 1D}. Consider a sequence of eigenpairs $(\lambda_{m},\bm u_{m})$ of $M_m$ so that $\lambda_{m} \to \lambda\in\R$ as $m\to\infty$.
    Then, $\lambda = f(e^{\i\alpha_0})$ with
    \begin{align*}
    \alpha_0 = \pm \lim_{m\to\infty} \mathcal{Q}_m(\bm u_{m}).
    \end{align*}
\end{proposition}
\begin{proof}
    It is well-known that the eigenpairs of $T_m(f)$ are given by
    \begin{align*}
        \lambda_{(m,s)} &= a_0 + 2 a_1\cos\left(\frac{s\pi}{m+1}\right), \quad 1\leq s\leq m,\\
        \bm u_{(m,s)}^{q} &=\kappa_m \sin\left(q\frac{s\pi}{m+1}\right), \quad 1\leq q,s\leq m,
    \end{align*}
    where $\kappa_m$ ensures that $\bm u_{(m,s)}$ is normalised.
    Denoting 
    \begin{align*}
    g_s:[0,1)&\to\R\\
    x&\mapsto \sin(\pi x s),\quad s\in\N,
    \end{align*}
    a straightforward computation yields $g_s = \sum_{n\in\Z} g_{s,n}e^{\i n x}$ with 
    \begin{align}
    \label{eq: gsn 2n}
        g_{s,n} = \begin{dcases}
        -\frac{1}{2\i}, &n=\pm r,\\
        0& else, 
        \end{dcases}
    \end{align}
    if $s=2r$ or
    \begin{align}
    \label{eq: gsn 2n+1}
        g_{s,n} = \frac{1}{\pi}\left(\frac{1}{2(r-n)+1}+\frac{1}{2(r+n)+1}\right)
    \end{align}
    if $s=2r+1$. In the first case $s=2r$, we trivially get $\mathcal{Q}_m(\bm u_{(m,s)})=\pi s/m$. On the other hand,  for $s=2r+1$, $\vert g_{s,n}\vert^2$ has a peak for $n=\pm r$ and decays quadratically from there. It follows therefore from a bounding argument that $\lim_{m\to\infty} \mathcal{Q}_m(\bm u_{(m,s)})=\pi \lim_{m\to\infty}s/m$ in this case.

    On the other hand, we remark that since $a_{-1}=a_1$, $f$ takes the form $f(e^{i\theta})=a_0 + 2 a_1\cos(\theta)$ and
    \begin{align*}
        \lambda_{(m,s)}\to a_0+2a_1\cos(\alpha_0) = f(e^{\i\alpha_0}),
    \end{align*}
    with $\alpha_0 = \pi \lim_{m\to\infty} s/m$, which concludes the proof for $T_m(f)$. For $C_m$, one may observe that, using the explicit formulas from \cite{ammari.barandun.ea2024Mathematical} and the linearity of $\mathcal{F}$, the argument presented above can be repeated as is.
\end{proof}
We present in \cref{fig:DFT_1D} the band structure reconstruction for systems of dimers with only nearest neighbour interactions. The reader may notice that every even quasiperiodicity is exactly reconstructed already in the $m=20$, while for the odd quasiperiodicities, we need larger system sizes to accurately reconstruct them.  This phenomenon is due to the difference between the decay patterns of \cref{eq: gsn 2n} and \cref{eq: gsn 2n+1}.
\begin{figure}[h]
    
    \centering
    \begin{subfigure}[t]{0.32\textwidth}
        \centering
        \includegraphics[width=\textwidth]{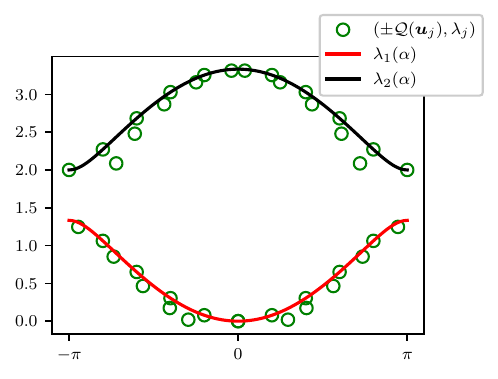}
    \caption{$m=20$}\end{subfigure}\hfill\begin{subfigure}[t]{0.32\textwidth}
        \centering
        \includegraphics[width=\textwidth]{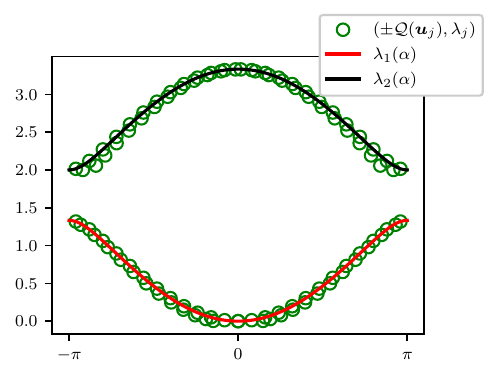}
        \caption{$m=50$}
    \end{subfigure}\hfill
    \begin{subfigure}[t]{0.32\textwidth}
        \centering
        \includegraphics[width=\textwidth]{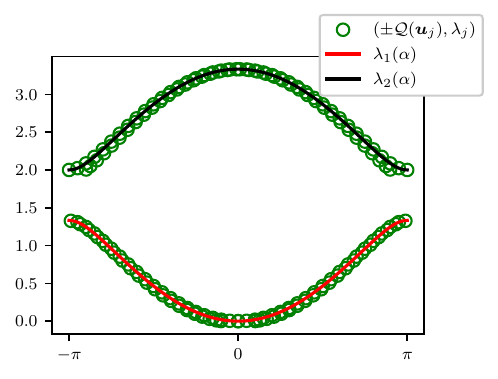}
        \caption{$m=80$}
    \end{subfigure}
    \caption{Band structure reconstruction for systems with only nearest neighbour interactions (eigenpairs of $\mathcal{C}_m$ from \eqref{eq: form cap mat 1D}). The solid lines show the traces of the symbol's eigenvalues (\emph{i.e.}, the actual periodic bands) while the green circles show the reconstructed discrete bands.}
    \label{fig:DFT_1D}
\end{figure}
\subsection{Systems of resonators with long-range interactions}
In \cref{subsec: 1D periodic monomers}, we have presented for the convenience of the reader the simplest possible system and an explicit proof for the reconstruction of the band structure though finite systems using the truncated Floquet-Bloch transform. Now, we will present our results in full generality.
\begin{theorem}\label{thm: TFT general}
     Let $f(z)=\sum_{n\in\Z}a_n z^n\in\setsymbols$ be such that $a_n = \mathcal{O}(n^{-p})$ for some $p>1$. Consider a sequence of eigenpairs $(\lambda_{m},\bm u_{m})$ of $T_{mk}(f)$ so that $\lambda_{m} \to \lambda\in\R$ as $m\to\infty$ with $\lambda=\lambda_p(e^{\i\alpha_0}) \in \sigma_{ess}(T(f))$ for $e^{\i\alpha_0}\neq \pm 1$ and $\norm{\bm u_m}=1$.
    Then 
    \begin{align*}
    \alpha_0 = \pm \lim_{m\to\infty} \mathcal{Q}_m(\bm u_{m}).
    \end{align*}
\end{theorem}
We remark that the assumption on the limit of $\lambda_{m}$ and $\bm u_{m}$ is rather weak. A well-established result, see \cref{prop: spectrum of block toeplitz is intervals and finite number of points}, states that the spectrum of a Toeplitz operator is composed out of the essential spectrum and at most a site of finitely many isolated eigenvalues. 
\begin{proof}[Proof of \cref{thm: TFT general}]
    The proof is now the combination of various steps prepared in the previous sections. Let $\epsilon>0$ and consider an eigenpair $T_m(f)\bm u_m = \lambda_m \bm u_m$. Then, by the assumption on the decay of the Fourier coefficients of $f$, we can approximate $T_m(f)$ by $T_m(f^{[r]})$ for some $r$. Specifically,  \cref{lemma: eve of full are pseudoeve of banded} holds for $f$ and thus we may choose $m$ large enough so that 
    \begin{align*}
        \Vert T_m(f^{[r]})\bm u_m - \lambda_m \bm u_m \Vert < \epsilon.
    \end{align*}
    Note now that by \cref{cor:general delocalization}, $\bm u_m$ must be delocalised, i.e. $\Vert \bm u_m\Vert_\infty \to 0$. Furthermore, $(C_m(f^{[r]})$ and $T_m(f^{[r]}))$ differ only in a fixed number, say $K_1$, of entries (independent of $m$) uniformly bounded by $K_2$. By possibly increasing $m$, we can assume that $K_1K_2\Vert \bm u_m\Vert_\infty<\epsilon$. Then,
    \begin{align*}
        & \Vert C_m(f^{[r]}) \bm u_m - \lambda_m \bm u_m \Vert = \Vert (C_m(f^{[r]}) - T_m(f^{[r]}))\bm u_m + T_m(f^{[r]})\bm u_m - \lambda_m\bm u_m \Vert \\&\leq \Vert (C_m(f^{[r]}) - T_m(f^{[r]})) \bm u_m\Vert + \Vert T_m(f^{[r]})\bm u_m - \lambda_m\bm u_m \Vert < K_1K_2\Vert \bm u_m\Vert_\infty + \varepsilon < 2\epsilon.
    \end{align*}
    As a last step, we can apply \cref{cor: Q recovers quasifrequency in circulant} to conclude the proof.
\end{proof}

We elucidate the result of \cref{thm: TFT general} with two applications. The first one, presented in \cref{fig: band reconstructed 3D cap mat}, reconstructs the band structure of the capacitance matrix associated to a one-dimensional chain of $100$ dimers in three dimensions. This should be compared to \cite[Figure 4]{ammari.davies.ea2023Spectral}.
\begin{figure}[!h]
    \includegraphics[width=0.75\textwidth]{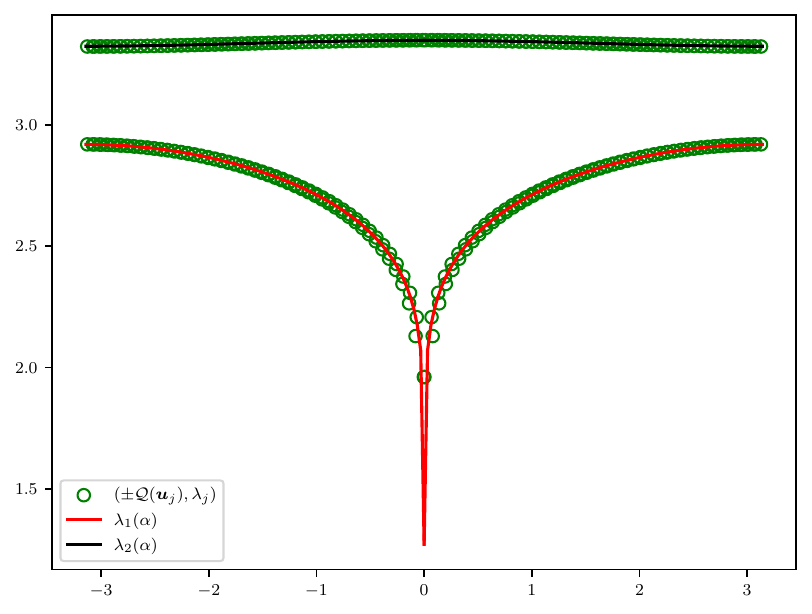}
    \caption{Band reconstruction for a three-dimensional dimer system composed of 100 resonators. Fourier coefficient of the matrix of interest decay slowly with approximately $\mathcal{O}(n^{-1.2})$. Nevertheless, the band structure is perfectly reconstructed.
    }
    \label{fig: band reconstructed 3D cap mat}
\end{figure}
\cref{fig: exponentially decaying symbol} shows instead the band reconstruction of $T_m(f)$ for
\begin{align}
\label{eq: exp decaying symbol}
    f(z) = \sum_{p\in \Z} \frac{-1}{2^{\vert p\vert}} z^p.
\end{align}
We observe a very fast reconstruction of the band structure, which is distinguishable already for $N=10$ and very well approximated for $N=30$.
\begin{figure}
\begin{subfigure}[t]{0.48\textwidth}
    \centering
    \includegraphics[width=\textwidth]{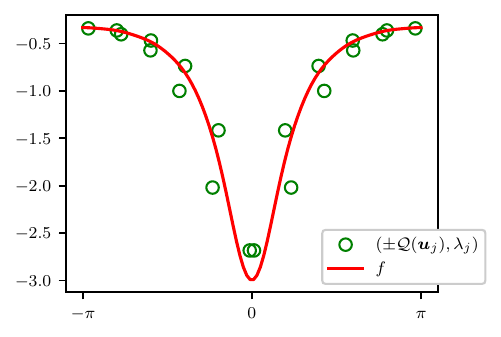}
    \caption{$N=10$}
\end{subfigure}\hfill\begin{subfigure}[t]{0.48\textwidth}
    \centering
    \includegraphics[width=\textwidth]{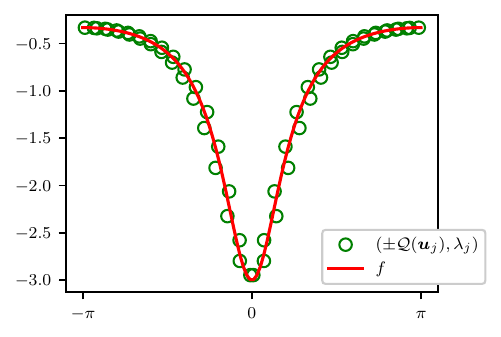}
    \caption{$N=30$}
\end{subfigure}
    \caption{Band reconstruction for $T_m(f)$ with $f$ as in \eqref{eq: exp decaying symbol} with exponentially decaying Fourier coefficients.}
    \label{fig: exponentially decaying symbol}
\end{figure}

\section{Aperiodic structures} \label{sec:6}
The Floquet-Bloch transform is a well-established tool to analyse periodic systems. It is, therefore, no wonder that its truncated counterpart works so well for periodic vectors as proved in the previous section. In this section, we show that its power extends far further by applying it to aperiodic structures, which are not encompassed by the standard Floquet-Bloch transform. We consider in particular three types of defected structures: SSH-like structures, dislocated structures, and structures with compact defects. All of these structures arise naturally from a periodic model (here a system of dimers). We will show that the truncated Floquet-Bloch transform recovers the band structure of the model and the localised eigenmodes inside the bandgap. 
\subsection{SSH-like structures}
In this subsection, we consider a model as the one studied in \cite{ammari.barandun.ea2024Exponentially}. This model is an SSH model, which supports one localised mode. The corresponding capacitance matrix is
given by 
\setcounter{MaxMatrixCols}{20}
\begin{align}
\label{eq: strucutre capacitance matrix ssh}
    \mathcal{C}_{\text{dSSH}} = \begin{pNiceMatrix}
    \Block[draw,fill=blue!40,rounded-corners]{7-7}{} \tilde{\alpha} & \beta_{1} &&&&&&&&&\\
\beta_{1} & \alpha & \beta_{2}  &&&&&&&&\\
& \beta_{2} & \alpha & \beta_{1}  &&&&&&&\\
       && \ddots     & \ddots     & \ddots     &&&&&&\\
       &&& \beta_{2} & \alpha & \beta_{1}  &&&&&\\
       &&&& \beta_{1} & \alpha & \beta_{2}  &&&&\\
       &&&&& \beta_{2} & \Block[draw,fill=red!40,rounded-corners]{7-7}{}\eta & \beta_{2}  &&&\\
       &&&&&& \beta_{2} & \alpha & \beta_{1}  &&\\
       &&&&& && \beta_{1} & \alpha & \beta_{2}  &\\
       &&&&&&&& \ddots     & \ddots & \ddots     &\\
       &&&&&&&&& \beta_{1} & \alpha & \beta_{2}  \\
       &&&&&&&&& & \beta_{2} & \alpha & \beta_{1}  \\
       &&&&&&&&&&& \beta_{1} & \tilde{\alpha}
    \end{pNiceMatrix}.
\end{align}
In \cite{ammari.barandun.ea2024Exponentially}, it shown that there exists exactly one eigenvector $\bm v\in \R^{4m+1}$ of $\mathcal{C}_{\text{dSSH}}\in \R^{4m+1\times 4m+1}$ which is exponentially localised, loosely meaning 
\begin{align}
\label{eq: eig dssh exp}
  \bm v^{(\vert 2m-j\vert)} \approx Ae^{-Bj}, \quad A \in \R, B >0,  
\end{align}
while all other eigenvectors $\bm u\in \R^{4m+1}$ behave like
\begin{align}
    \label{eq: eig dssh deloc}
    \bm u^{(\vert 2m-j\vert)} \approx A_1\sin(B_1 j)+A_2\cos(B_2 j), \quad A_1,A_2,B_1,B_2\in\R.
\end{align}
We refer the reader to \cite[Proposition 6]{ammari.barandun.ea2024Exponentially} for a more precise statement on the eigenvectors of $\mathcal{C}_{\text{dSSH}}$. One may remark that eigenvectors of the form \eqref{eq: eig dssh deloc} are approximately the eigenvectors of the unperturbed structure (see \cite{ammari.barandun.ea2023Perturbed}). So, it is clear that for these eigenvectors the result of \cref{thm: TFT general} holds. On the other hand, \cref{eq: eig dssh exp} is localised and thus \cref{thm: TFT general} does not hold. In \cref{fig: main ssh}, we illustrate how the reconstruction of the band structure works. In particular, \cref{fig: TFBP loc and deloc eve} shows how localised and delocalised modes are different in front of the truncated Floquet-Bloch. This difference can be explained with the \emph{uncertainty principle}: a vector and its Fourier transform cannot both be localised. Remark that the eigenvectors of $\mathcal{C}_{\text{dSSH}}$ cannot directly be fed into $\mathcal{Q}$ as the latter expects a vector of even dimension, while the SSH structure has an odd number of resonators due to the defect. We consequently zero-pad the eigenvectors before applying $\mathcal{Q}$.

\begin{figure}[h]
    \centering
    \hfill\begin{subfigure}[t]{0.32\textwidth}
        \centering
        \includegraphics[width=\textwidth]{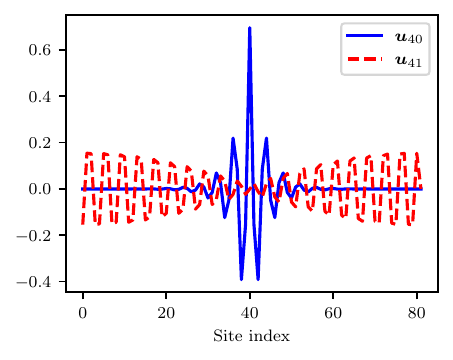}
        \caption{The localised eigenvector $\bm u_{40}$ and a delocalised eigenvector.}
        \label{fig: loc and deloc eve}
    \end{subfigure}\hfill
    \begin{subfigure}[t]{0.32\textwidth}
        \centering
        \includegraphics[width=\textwidth]{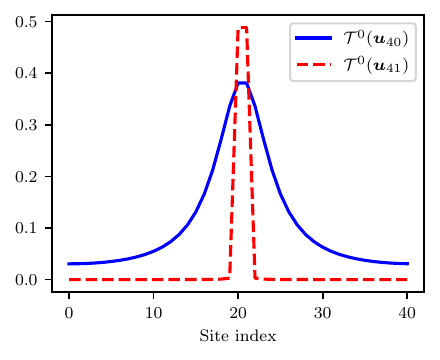}
        \caption{The truncated Floquet-Bloch projection for the eigenvectors of \cref{fig: loc and deloc eve}.}
        \label{fig: TFBP loc and deloc eve}
    \end{subfigure}\hfill\phantom{.}\\
    \begin{subfigure}[t]{0.5\textwidth}
        \centering
        \includegraphics[width=\textwidth]{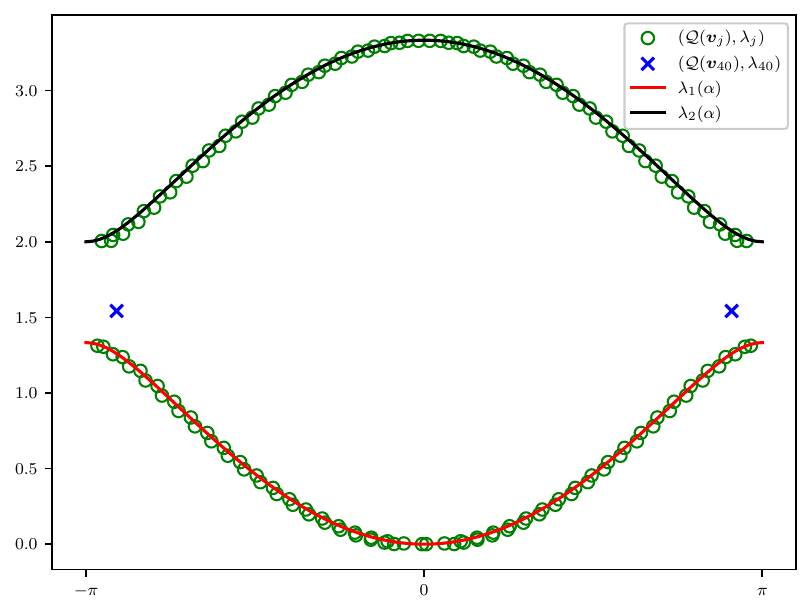}
        \caption{Reconstructed band structure and band structure of the periodic model.}
    \end{subfigure}
    \caption{The truncated Floquet-Bloch reconstructs the band structure associated to  \eqref{eq: strucutre capacitance matrix ssh}. Note that the band structure is the same as the one of the periodic structure from \cref{fig:DFT_1D} with the exception of the localised mode that stands out. Figures (\textsc{a}) and (\textsc{b}) illustrate the uncertainty principle.}
    \label{fig: main ssh}
\end{figure}
\subsection{Dislocated structures}
The second aperiodic system we consider is the one of a dislocated structure. This kind of structures, which have been widely studied \cite{ammari.davies.ea2022Robust,hempel.kohlmann2011variational,drouot.fefferman.ea2020Defect}, are composed by a dimer structure which gets separated into two structures with a distance $d$ between them; see \cref{fig: disolcated setup}. It is well-known that this dislocation generates a mid gap frequency associated to a topologically protected localised interface mode. We show in \cref{fig: dislocated band structure} that exactly as for SSH structures, also for dislocated structures the discrete Floquet-Bloch transform perfectly recovers the band structure of the originally periodic system and isolates the localised eigenmode.

\begin{figure}[h]
    \centering
    \begin{subfigure}[t]{0.58\textwidth}
        \centering
        \raisebox{1.4cm}{
\begin{tikzpicture}

\def\d{1.5} 
\def\v{-1} 

\node[circle, draw] at (0, 0) {};
\node[circle, draw] at (.5, 0) {};
\node[circle, draw] at (1.5, 0) {};
\node[circle, draw] at (2, 0) {};
\node[circle, draw] at (3, 0) {};
\node[circle, draw] at (3.5, 0) {};
\node[circle, draw] at (4.5, 0) {};
\node[circle, draw] at (5, 0) {};
\node[circle, draw] at (6, 0) {};
\node[circle, draw] at (6.5, 0) {};

    \begin{scope}[shift={(-.5,0)}]
\node[circle, draw] at (0, \v) {};
\node[circle, draw] at (.5, \v) {};
\node[circle, draw] at (1.5, \v) {};
\node[circle, draw] at (2, \v) {};
\node[circle, draw] at (3, \v) {};
    \end{scope}
    \begin{scope}[shift={(+.5,0)}]
\node[circle, draw] at (3.5, \v) {};
\node[circle, draw] at (4.5, \v) {};
\node[circle, draw] at (5, \v) {};
\node[circle, draw] at (6, \v) {};
\node[circle, draw] at (6.5, \v) {};
\end{scope}
\draw[->] (3.25,0) to[out=-90, in=0] (2.75,\v);
\draw[->] (3.25,0) to[out=-90, in=180] (3.75,\v);
\draw[|-|] (2.75,\v-.2) -- (3.75,\v-.2);
\node[anchor=south] (d) at (3.25,\v-.2) {$d$};
\end{tikzpicture}}
        \caption{Dislocation occurs in a system of dimers by increasing the intra-dimer spacing to an arbitrary distance $d$.}
        \label{fig: disolcated setup}
    \end{subfigure}\hfill
    \begin{subfigure}[t]{0.4\textwidth}
        \centering
        \includegraphics[width=\textwidth]{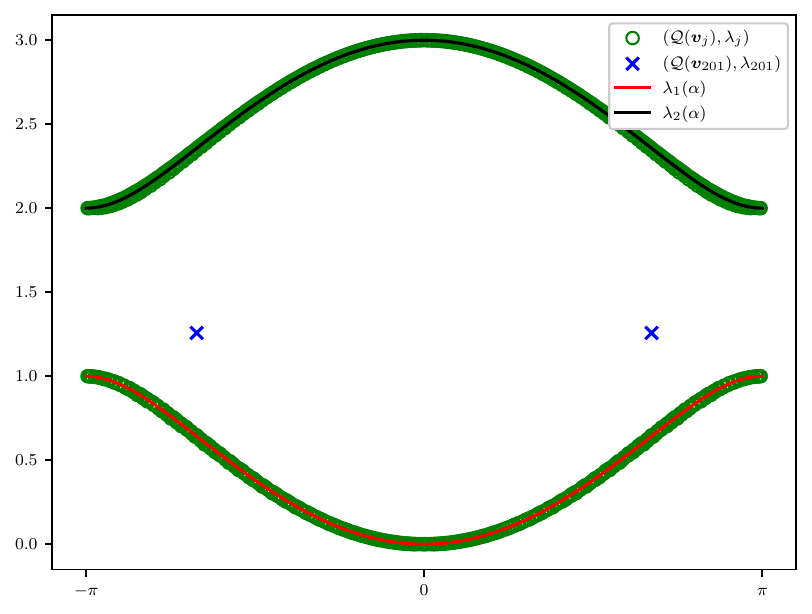}
        \caption{The truncated Floquet-Bloch transform effectively recovers the band structure of the original (periodic) system, here shown as red and black solid lines. Blue crosses show the quasiperiodicity and eigenvalue of the localised mode.}
        \label{fig: dislocated band structure}
    \end{subfigure}
    \caption{The truncated Floquet-Bloch transform can be applied to a dislocated system to recover the band structure of the underlying periodic system.}
    \label{fig: dislocation full figure}
\end{figure}
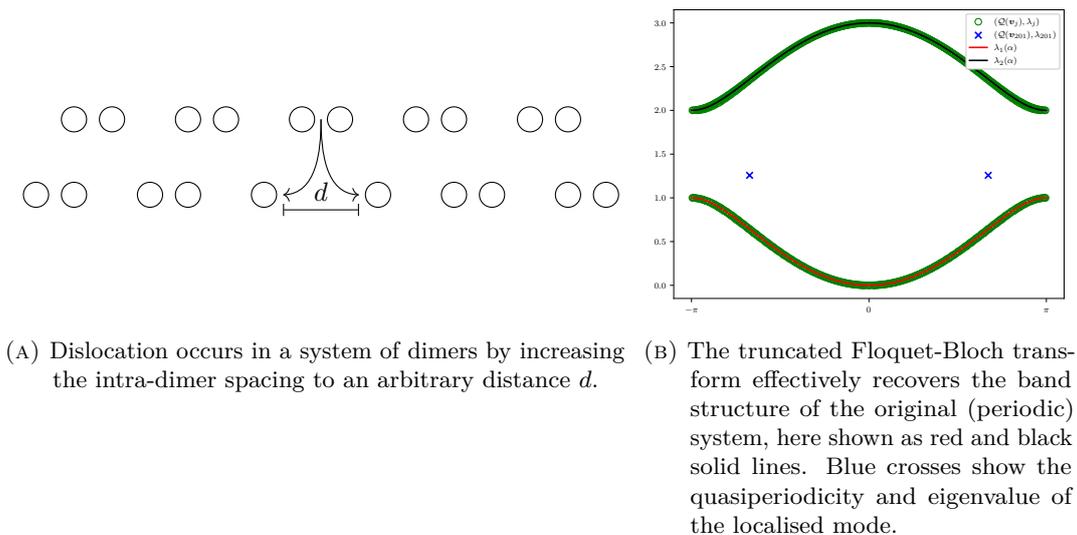

\subsection{Compact perturbations}
The last aperiodic system we consider is the one with compact perturbations obtained here by perturbing the material parameters of one resonator. If $C$ is the matrix associated to the unperturbed system, then in order to analyse the perturbed system we look at eigenvalues of $B C$, 
where 
\begin{align}
    \label{eq: defect matrix}
    B=\diag(1,\dots,1,1+\delta,1,\dots,1)
\end{align} with $\delta\in\R$. It is clear that the method introduced in \cite{ammari.davies.ea2024Anderson} can be used to analyse the eigenvectors of $BC$, showing that the eigenvectors associated to eigenvalues in the spectral bulk are delocalised and are approximately the ones of the defectless structure. However, the eigenvectors associated to eigenvalues in the gap are exponentially decaying. A full and rigorous analysis of the existence and uniqueness of eigenvalues in the gap is not in the scope of this paper but could be performed using the same tools as those in \cite{ammari.barandun.ea2024Exponentially}. Here, we are content with simply analysing the band structure, which we do in \cref{fig: compact defect}. In \cref{fig: eve negative compact defect,fig: negative compact defect}, one notices that in the presence of perturbations that do not generate eigenvalues in the gap, all eigenvectors are delocalised and the reconstruction of the periodic band structure is flawless. In \cref{fig: eve positiv compact defect,fig: positiv compact defect}, we look at perturbations that generate eigenvalues in the gap. Here the reconstruction of the band is successful as well, but the pairs $(\mathcal{Q}(\bm u), \lambda)$ associated to localised eigenvectors stand out again, due to the delocalised nature of the discrete Fourier transform of these vectors.
\begin{figure}[h]
    \centering
    \begin{subfigure}[t]{0.48\textwidth}
        \centering
        \includegraphics[width=\textwidth]{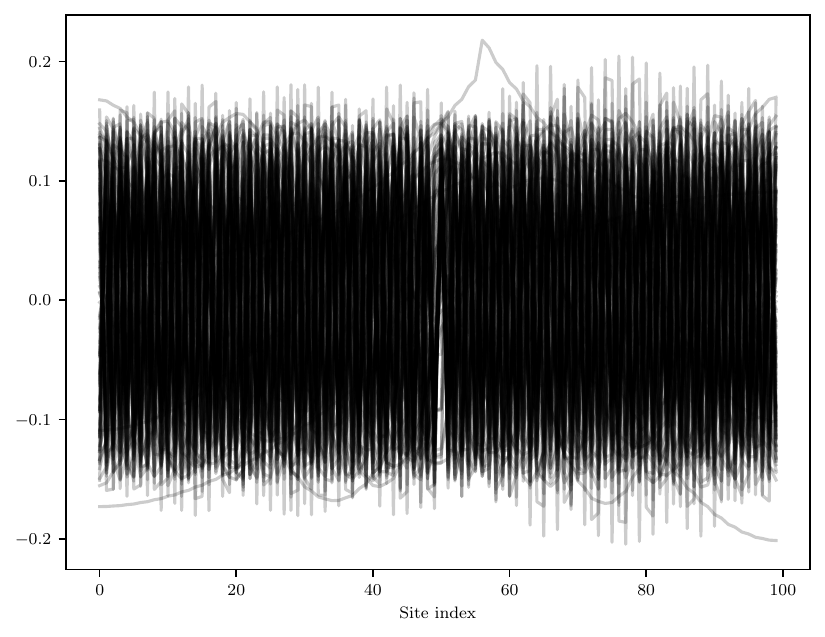}
        \caption{Negative defect, eigenmodes superimposed.}
        \label{fig: eve negative compact defect}
    \end{subfigure}\hfill
    \begin{subfigure}[t]{0.48\textwidth}
        \centering
        \includegraphics[width=\textwidth]{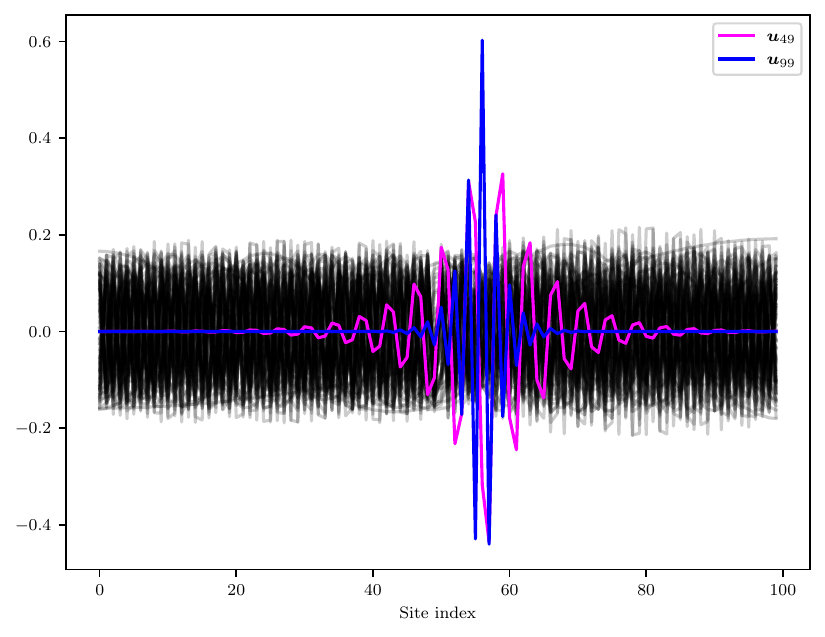}
        \caption{Positive defect, eigenmodes superimposed. Localised eigenmodes are highlighted with colours.}
        \label{fig: eve positiv compact defect}
    \end{subfigure}
    \begin{subfigure}[t]{0.48\textwidth}
        \centering
        \includegraphics[width=\textwidth]{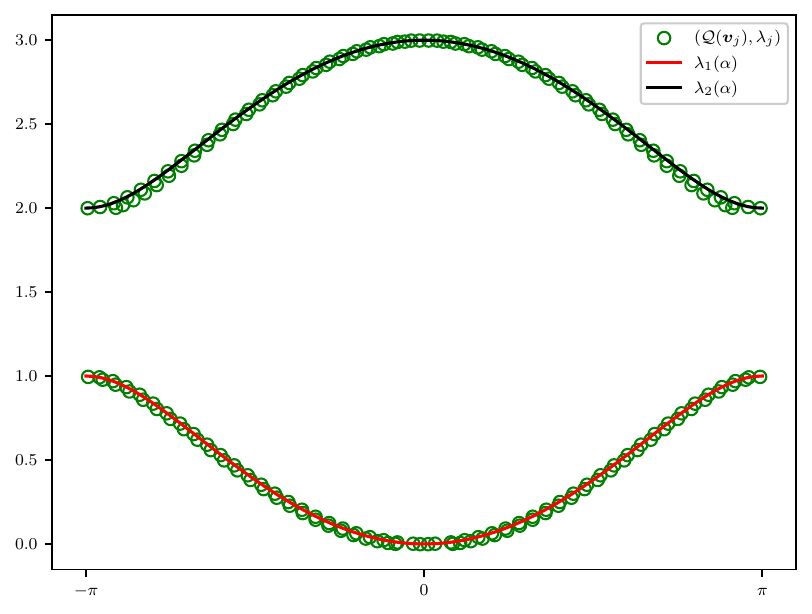}
        \caption{Negative defect, reconstructed band structure.}
        \label{fig: negative compact defect}
    \end{subfigure}\hfill
    \begin{subfigure}[t]{0.48\textwidth}
        \centering
        \includegraphics[width=\textwidth]{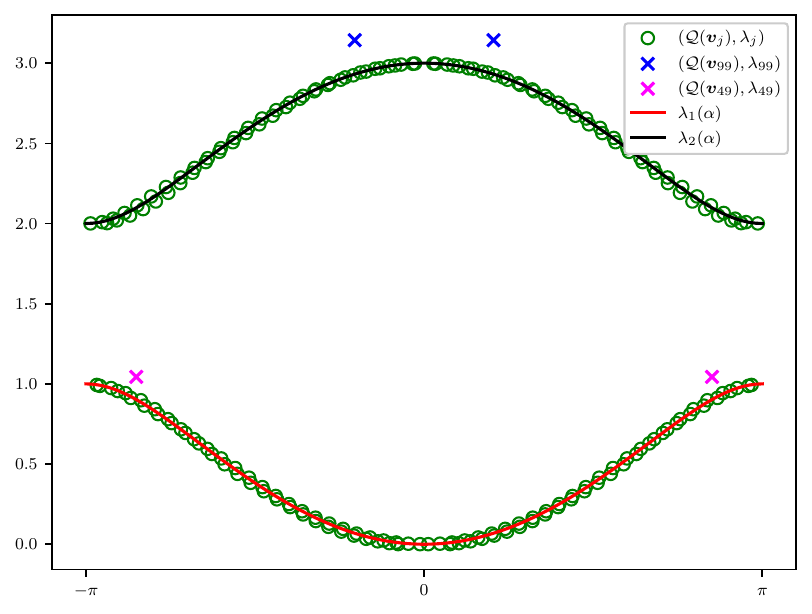}
        \caption{Positive defect, reconstructed band structure. Eigenvalues associated to localised eigenmodes are highlighted with the same colours as in (\textsc{b}).}
        \label{fig: positiv compact defect}
    \end{subfigure}
    \caption{Band reconstruction for a compact defect of the type of \eqref{eq: defect matrix}. For negative defects, no localised modes are generated and the band is recovered flawlessly. For positive defects, localised modes are generated and the corresponding eigenvalues stand out in the reconstructed band structure.}
    \label{fig: compact defect}
\end{figure}

\section{Concluding remarks} \label{sec:7}
In this paper, we have provided the mathematical foundations of the truncated Floquet-Bloch transform. We have applied the proposed method for finite systems of resonators to characterise both delocalised and localised eigenmodes. We have illustrated its efficiency and accuracy in a variety of numerical examples. Following this line of research, it would be very interesting to generalise the method and the results obtained in this paper to more general disordered structures than those studied in this paper. Another very interesting problem is 
to apply the truncated Floquet-transform to non-Hermitian systems in particular those exhibiting a skin effect (accumulation of the eigenmodes at one edge of the structure)
based on our previous results in \cite{ammari.barandun.ea2024Generalised,ammari.barandun.ea2024Mathematical}.

\section*{Acknowledgements}
The authors would like to thank Ping Liu for insightful discussions. They also thank  Bryn Davies and Erik Hiltunen for kindly sharing the code to generate the matrix associated to this physical system shown in \cref{fig: band reconstructed 3D cap mat}.  The work of AU was supported by Swiss National Science Foundation grant number 200021--200307.

\section*{Code availability}
The software used to produce the numerical results in this work is openly available at \\ \href{https://doi.org/10.5281/zenodo.13949362}{https://doi.org/10.5281/zenodo.13949362}.

\appendix
\section{Technical results}
We present here a series of rather straightforward but helpful results that are used in the main text.
The following well-known result comes from \cite[Section 6.1]{bottcher.silbermann1999Introduction}.
\begin{lemma}\label{lemma: norm of toeplitz operator via symbol}
    Let $f\in C_{k\times k}$. Then, we have 
    \begin{align*}
        \Vert T(f)\Vert = \Vert f\Vert_\infty.
    \end{align*}
\end{lemma}
The next results are from \cite[Section 6.4]{bottcher.silbermann1999Introduction}. 
\begin{lemma}\label{lemma: norm of toeplitz sections converges}
    Let $f\in C_{k\times k}$. Then, we have 
    \begin{align*}
        \Vert T_m(f)\Vert \to \Vert T(f)\Vert \quad \text{as}\quad m\to\infty .
    \end{align*}
\end{lemma}
\begin{lemma}\label{lemma: strong convergence of sections}
Let $f\in\setsymbols$. Then $T_m(f)\to T(f)$ strongly in $\ell^2(\N)$, that is, for all $x\in \ell^2(\N)$ and for all $\epsilon>0$, there exists an $N_0\in\N$ so that for all $m>N_0$, 
\begin{align*}
    \Vert T_m x - Tx \Vert < \epsilon.
\end{align*}
\end{lemma}
\begin{proof}
    For $x\in \ell^2(\N)$, we can estimate
    \begin{align*}
        \Vert T_m x - Tx \Vert^2 = \sum_{j>m} \Vert T_m - T \Vert^2 \vert x_j\vert^2 \leq  3\Vert T \Vert \sum_{j>m}\vert x_j\vert^2 ,
    \end{align*}
    which converges as $x\in \ell^2(\N)$.
\end{proof}

\begin{lemma}\label{lemma: eve of full are pseudoeve of banded}
Consider a symbol $f\in\setsymbols$ such that the Fourier coefficients $a_k$ of $f$ decay as
$$
\Vert a_k \Vert_\infty = \mathcal{O}(k^{-p})
$$
for some $p>1$ and let $\epsilon>0$. Then, there exists an $r_0$ and an $N_0\in \N$ such that for any $N>N_0$ and $r>r_0$ if $T_N(f) u_N = \lambda_N u_N$, then the following estimate holds:
\begin{align*}
    \Vert T_N(\rapprox{f}{r})u_N - \lambda_N u_N \Vert < \epsilon.
\end{align*}
\end{lemma}
\begin{proof}
    One notices that
    \begin{align*}
        \Vert T_N(\rapprox{f}{r})u_N - \lambda_N u_N \Vert &= \Vert (T_N(f) - T_N(\rapprox{f}{r}))u_N\Vert \\&\leq \Vert T_N(f) - T_N(\rapprox{f}{r})\Vert = \Vert f - \rapprox{f}{r} \Vert_\infty,
    \end{align*}
    where in the last equality we have used \cref{lemma: norm of toeplitz operator via symbol}. Therefore, 
    \begin{align*}
        \Vert f - \rapprox{f}{r} \Vert_\infty = \sup_{t\in \mathbb{T}^1} \Vert f(t) - \rapprox{f}{r}(t)\Vert_\infty\leq \sum_{k = r}^{N} \Vert a_k \Vert_\infty \leq \kappa\frac{N^{1-p}-(r-1)^{1-p}}{1-p}
    \end{align*}
    for some $\kappa>0$. This concludes the proof.
\end{proof}

The following result from \cite{feldman1993Finiteness} is also of use to us.
\begin{proposition}\label{prop: spectrum of block toeplitz is intervals and finite number of points}
    Let $f(z):\mathbb{T}^1\to\C^{k\times k}$ be a Hermitian matrix function holomorphic on $\mathbb{T}^1$ and assume that for any $\lambda$ the function $\det(f(z)-\lambda)$ does not vanish identically with respect to $z$. Then, the spectrum of the Toeplitz operator $T(f)$ consists of a finite number of intervals belonging to the essential spectrum and at most a finite set of isolated eigenvalues.
\end{proposition}

\printbibliography

\end{document}